%% file: do_allmain.tex
\documentclass[a4paper,11pt, notitlepage]{article}
\usepackage[margin=1.0in]{geometry}
\usepackage{times}
\usepackage[T1]{fontenc}
\usepackage[utf8]{inputenc}
\usepackage{listings}
\usepackage{geometry}
\usepackage{graphicx}
\usepackage{amsfonts}
\usepackage{tikz}
\usepackage{hyperref}
\usepackage{amsmath, amsthm}
\usepackage[linesnumbered, algoruled]{algorithm2e}
\usepackage{color}

\newtheorem{lemma}{Lemma}
\newtheorem{theorem}{Theorem}
\newtheorem{fact}{Fact}
\newtheorem{corollary}{Corollary}
\newtheorem{definition}{Definition}
\newtheorem{claim}{Claim}

\newcommand{\cO}{{\mathcal O}}

\newcommand{\cM}{{\mathcal M}}

\newcommand{\cF}{{\mathcal F}}

\newcommand{\safba}{{Strongly-Adaptive $f$-Bounded}}

\newcommand{\walba}{{Weakly-Adaptive Linearly-Bounded}}

\newcommand{\GT}{\textsc{Groups-Together}}

\newcommand{\DA}{\textit{Do-All}}

\newcommand{\Grubtech}{\textsc{GrubTEch}}
\newcommand{\Gilet}{\textsc{GILET}}
\newcommand{\Robal}{\textsc{ROBAL}}

\newcommand{\remove}[1]{}

\newcommand{\jm}[1]{#1}
\newcommand{\dk}[1]{#1}

\newcommand{\mj}[1]{#1}
\newcommand{\jmii}[1]{#1}

\usepackage{titlesec}
\usepackage{lipsum}

\titlespacing{\paragraph}{%
  0pt}{
  0.1\baselineskip}{
  1em}

\begin{document}

\title{Ordered and Delayed Adversaries\\ 
and How to Work against Them on a Shared Channel}

\author{Marek Klonowski\footnotemark[1] \and 
Dariusz R. Kowalski\footnotemark[2] \and 
Jarosław Mirek\footnotemark[2]}

\footnotetext[1]{
			Faculty of Fundamental Problems of Technology,
			Wrocław University of Technology, \\
			Wybrzeże Wyspiańskiego 27, 50-370 Wrocław, Poland.\\
Email: \texttt{Marek.Klonowski@pwr.edu.pl}
}

\footnotetext[2]{
                	Department of Computer Science,
                	University of Liverpool,\\
                	Ashton Building, Ashton Street,
               		Liverpool L69 3BX, UK.\\
Email: \texttt{\{D.Kowalski,J.Mirek\}@liverpool.ac.uk}
}
\footnotetext{Partially supported by Polish National Science Center 
grant 2015/17/B/ST6/01897.}
\date{}

\maketitle

\begin{abstract}
\input{abstract}
\end{abstract}

\setcounter{page}{0}
\thispagestyle{empty}

\pagebreak

\input{introduction}
\input{model}

\input{algtools}
\input{robal}
\input{grubtech}
\input{grubtechpartord}

\input{gilet}

\input{beeping}
\input{conclusions}

\section*{Acknowledgements}

We would like to thank anonymous reviewers, whose detailed and accurate comments allowed to improve significantly the presentation of results.

\input{bibliography.bbl}

\end{document}

%% file: abstract.tex
\dk{An execution of a distributed algorithm is often seen}  
as a game between
the algorithm and a conceptual adversary
causing specific distractions \jm{to the computation}.
In this work we define a class of \textit{ordered \dk{adaptive} adversaries}, which 
cause distractions \dk{--- in particular crashes --- online} according to some partial order 
\dk{of the participating stations}, which is fixed by the adversary
before the execution. 
\dk{We distinguish: 
\textit{Linearly-Ordered adversary}, which is restricted by some pre-defined
{\em linear} order of (potentially) crashing stations;
\textit{Anti-Chain-Ordered adversary}, 
previously known as the \textit{Weakly-Adaptive adversary},
which is restricted by some pre-defined {\em set} of crash-prone stations
(it can be seen as an ordered adversary with the order being an anti-chain,
i.e., a collection of incomparable elements,
consisting of these stations);
\textit{$k$-Thick-Ordered adversary} restricted by
{\em partial orders} of 
stations with a maximum anti-chain of size $k$.}
\dk{We initiate a study of} 
how they affect performance of algorithms.
For this purpose, we focus on the well-known Do-All problem of performing $t$ 
tasks 
\dk{by $p$ synchronous crash-prone stations communicating on a shared channel}.
The channel restricts communication by the fact that no message is delivered
to the \jm{operational} stations if more than one station transmits
at the same time.
%
\\
The question addressed in this work is how the ordered adversaries 
controlling crashes of stations influence work performance,
\dk{defined as the total number of available processor 
steps during the whole execution and introduced by Kanellakis and Shvartsman
in \cite{KS} in the context of Write-All algorithms.}
%
%
The first \dk{presented} algorithm solves the Do-All problem with work $\cO(t+p \sqrt{t}\log p)$
against the Linearly-Ordered adversary. 
Surprisingly, the upper bound on performance of this algorithm does not 
depend on the number of crashes $f$ and is close to the absolute lower bound
$\Omega(t+p\sqrt{t})$ proved in~\cite{CKL}.
Another algorithm is developed against the Weakly-Adaptive adversary.
Work done by this algorithm is $\mathcal{O}(t + p\sqrt{t} + p\min\left\{p/(p-f),t\right\}\log p )$,
which is close to the lower bound $\Omega(t + p\sqrt{t} + p\min\left\{p/(p-f),t\right\})$
proved in~\cite{CKL} \dk{and answers the open questions posed there.}
We generalize this result to the class of $k$-Thick-Ordered adversaries, 
in which case the work of the algorithm is bounded by 
$\mathcal{O}(t + p\sqrt{t} + p\min\left\{p/(p-f),k,t\right\}\log p )$.
We complement this result by \dk{proving} the almost matching lower bound
$\Omega(t + p\sqrt{t} + p\min\left\{p/(p-f),k,t\right\})$.
\\
\dk{Independently from}
the results for the ordered adversaries, we consider a class
of \textit{delayed adaptive adversaries}, 
which could see random choices with
some delay. 
We present an algorithm that works efficiently against the \textit{$1$-RD} 
adversary, which could see random choices of stations with one round delay,
achieving close to optimal $\cO(t+p \sqrt{t}\log^{2} p)$ work complexity. 
This shows that restricting 
\dk{the adversary by not allowing it to react on random decisions
immediately makes it significantly weaker, in the sense that there is an algorithm
achieving (almost) optimal work performance.}
%
%

\remove{
This leads to several randomized algorithms that we created.

To analyze performance of our algorithms we adopt the 
most popular complexity measure in the literature to this problem, which is work, defined as the number of the available processor steps
during the whole computation; we also comment on
the other two popular measures such as time and 
transmission energy.

We consider only \textit{reliable} algorithms that perform all the tasks  
as long as at least one station remains operational.

; we also comment on
the other two popular measures such as time and 
transmission energy.

Finally we give some remarks on the optimal solutions for the beeping model, what makes the problem solved for such setting.}

\noindent
{\bf Keywords:} Performing tasks, Do-All, Shared channel, Multiple-access channel, 
Ordered adversaries, Delayed adversaries, Crash failures, 
Distributed algorithms, Randomized Algorithms, Work complexity, Time complexity,
Transmission energy complexity. 
%

%% file: introduction.tex
\section{Introduction}
\label{introduction}

We consider the problem of performing $t$ similar and independent tasks 
in a distributed system prone to processor crashes. 
This problem, called \DA, was introduced by 
Dwork et al.~\cite{DHW} 
in the context of a message-passing system.
Over the years the \DA\ problem became a pillar of distributed computing
and has been studied widely from different perspectives \cite{GSbook2,GSbook}.

The distributed system studied in our paper is based on communication
over a shared channel, also called a multiple-access channel, and was first 
studied in the context of \DA\ problem by Chlebus et al.~\cite{CKL}. 
%
\dk{The communication channel  and the $p$ crash-prone processors, also called stations, connected to it
are synchronous.}  
\jm{They have access to a global clock,} 
which defines the same rounds for all \jm{operational} stations.
A message sent by a station at a round is received by all \dk{operational
(i.e., not yet crashed)} stations
only if it is the only transmitter in this round; we call such a transmission
successful.
Otherwise, unless stated differently,
\remove{
\footnote{%
There are two places in the paper referring to a channel with slightly enriched
feedback. One such place is when we recall algorithms
working in the setting with collision detection, i.e., when each station
receives additional channel feedback whether there was more than one transmitter
at the round, and show how to emulate such algorithms on a simple channel 
without such capability. Another part is at the end of this work and provides
some facts about the beeping channel, i.e., a channel when the only 
information from the channel is whether there was at least one transmitter at the round.
}
}
we assume that no station receives 
any meaningful feedback from the channel medium, except an acknowledgment 
of its successful transmission;
\dk{this setting is called \textit{without collision detection}, oppose to the
setting \textit{with collision detection} in which stations can distinguish between
no transmission and simultaneous transmission of at least two stations (we will
refer to the latter in few places of this work).}
%


Stations are prone to crash failures. 
Allowable patterns of failures are determined by abstract adversarial models.
Historically, the main distinction was between 
{\em adaptive} and {\em oblivious} adversaries;
the former can make decisions \dk{about failures} during the computation while the latter \dk{has to make all decisions prior the computation}.
Another characteristic of an adversary is its {\em size-boundedness}, 
or more specifically {\em $f$-boundedness}, if it may fail at most $f$~stations, for a parameter $0\le f< p$; a {\em linearly-bounded} adversary is simply a
$c\cdot p$-bounded adversary, for some constant $0<c<1$.
We introduce the notion of the ordered adaptive adversary, 
or simply ordered adversary,
which can crash stations \dk{online but} according to some preselected\footnote{%
\dk{Throughout the paper, by ``preselected'' we mean selected prior the execution
of a given algorithm,
as will be clarified later in the model description in Section~\ref{model}.}}
 order (unknown to
the algorithm) \dk{from a given class of partial orders, e.g., linear orders
(i.e., all elements are comparable),
anti-chains (i.e., no two elements are comparable), $k$-thick partial orders
(i.e., at most $k$ elements are incomparable).}
On the other hand, a strongly-adaptive adversary is not restricted
by any constraint other than being $f$-bounded for some $f<p$.
\remove{
A size-bounded adversary is weakly adaptive, if it needs to select a
subset of stations that might be failed prior to a start of an execution,
otherwise it is {\em strongly adaptive}.
}

Adversaries described by a partial order are interesting on their own right. 
To the best of our knowledge, such general adversaries
were not considered in literature so far and hence \dk{offer} novel \dk{framework} for 
\dk{evaluating performance} and 
better understanding of reliable distributed algorithms.
\dk{For instance, in hierarchical architectures, such as clouds or networks, a crash at 
an upper level may result in cutting off processors at lower levels, what could be seen
as a crash of lower levels from the system perspective.
Linear orders of crashes could be motivated by the fact that each station
has its own duration, unknown to the algorithm, and a crash of some station
means that all with smaller duration should crash as well.
Independent, e.g., located far away, systems could be seen as a set
of linearly ordered chains of crashes, each chain for a different independent
part of the system.
Furthermore, the study in this paper indicates that 
different partial orders restricting the adversary may require different
algorithm design and yield different complexity formulas.}

\dk{Another form of restricting adversarial power is to delay the effect of its actions
by a number of time steps --- we call such adversaries round-delayed. 
Such adaptive adversaries are motivated by various reactive attacking systems,
especially in analyzing security aspects. We show that this parameter can
influence performance of algorithms, independently of the partial-order restrictions
on the adversary.}

\dk{Due to the specific nature of the \DA\ problem,}
the most \jm{accurate} measure considered in the literature
is {\em work}, accounting the total number of available processor steps
in the computation. \dk{It was introduced by Kanellakis and Shvartsman in the context
of the related Write-All problem \cite{KS}.}
\jm{We assume that algorithms 
are {\em reliable}} in the sense that they must 
perform all the tasks for any pattern of crashes such that
at least one station remains operational in an execution.
Chlebus et al.~\cite{CKL} showed that $\Omega(t+p\sqrt{t})$ work is inevitable 
for any reliable algorithm, even randomized, even for the channel with collision detection 
(i.e., when \jm{operational} stations can recognize no transmission from at least two transmitting
stations in a round),
and even if no failure occurs.
This is the absolute lower bound on the work complexity of the \DA\ problem on a shared channel.
\dk{It is known that this bound can be achieved even by a deterministic
algorithm for channels with \textit{enhanced feedback}, such as collision detection,
cf.,~\cite{CKL}, and therefore such enhanced-feedback channels are no longer
interesting 
from perspective of reliable task performance.}
%
\textit{Our goal is to check how different classes of adversaries, especially those 
constrained by a partial order of crashes, influence work performance
of \DA\ algorithms on a simple shared channel with acknowledgments only.}

\subsection{Previous work}

The \DA\ problem  was introduced by Dwork, Halpern and Waarts~\cite{DHW}
in the context of a message-passing model with processor crashes.

Chlebus, Kowalski and Lingas \jm{(CKL)}~\cite{CKL} were the first who considered 
\DA\ in a multiple-access channel.
Apart from the absolute lower bound for work complexity, discussed earlier,
they also showed a deterministic algorithm matching this performance
in case of channel \textit{with} collision detection.
Regarding the channel without collision detection, they developed 
a deterministic solution 
\remove{and a Strongly-Adaptive adversary,} 
that is optimal for such weak channel with respect to the lower bound they proved
$\Omega(t + p\sqrt{t} + p\min\left\{f,t\right\})$. The lower bound
holds also for randomized algorithms against the strongly adaptive adversary,
that is, the adversary who can see random choices and react online,
which shows that randomization does not help against a strongly adaptive adversary.

\remove{Their results also consisted of a deterministic protocol for the Weakly-Adaptive adversary and a channel
without collision detection.} 
Furthermore, their paper contains a randomized solution that is efficient
against a weakly adaptive adversary who can fail only a constant fraction of stations.
A weakly adaptive adversary is such that it needs to select $f$ crash-prone 
processors in advance, based only on the knowledge of algorithm but without
any knowledge of random bits; then, during the execution, it can fail only
processors from that set. 
This algorithm matches the absolute lower bound on work. 
If the adversary is not linearly bounded, that is, $f<p$ could be arbitrary,
they only proved a lower bound of
$ \Omega(t + p\sqrt{t} + p\min\left\{\frac{p}{p-f},t\right\}) $.

Clementi, Monti and Silvestri~\cite{CMS} investigated \DA\ in the 
communication model of a multiple-access channel without collision detection.
They studied \textit{$F$-reliable protocols}, which are correct if the number 
of crashes is at most~$F$, for a parameter~$F<p$.
They obtained tight bounds on the time and work of $F$-reliable deterministic protocols.
In particular,  the bound on work shown in~\cite{CMS} is $\Theta(t+F\cdot\min \{t, F\})$.
In this paper, we consider protocols that are correct for \textit{any number} of crashes smaller than $p$, which is the same as $(p-1)$-reliability.
Moreover, the complexity bounds of our algorithms, for the channel without collision detection, are parametrized by the number $f$ of crashes that 
\textit{actually occur} in an execution. \jm{Results shown in \cite{CMS} also referred to the time perspective with a lower bound on time complexity}
equal $ \Omega\left(\frac{t}{p-F} + \min\left\{\frac{tF}{p}, F + \sqrt{t} \right\}\right)$. However the protocols 
make explicit use of the knowledge of $F$.
\mj{In this paper we give some remarks on time and energy complexity but, opposed to results in \cite{CMS}, those statements are correct for an
arbitrary $f$, which does not need to be known by the system (see details in Section \ref{timeenergy}).
Observe that, opposed to the worst case time complexity, the considered work complexity could be seen as an average processors time multiplied by the number of processors.}

\subsection{Related work}

\paragraph{Do-All problem.}

After the seminal work by Dwork, Halpern and Waarts~\cite{DHW},
the \DA\ problem was studied in a number of follow-up papers~\cite{CDS,CGKS,CK,DMY,GMY}
in the context of a message-passing model, in which every node 
can send a message to any subset of nodes in one round.
%
Dwork et al.~\cite{DHW} analyzed task-oriented work, 
in which each performance of a task contributes a unit to complexity, and the communication complexity defined as the number of point-to-point messages.
%
De Prisco, Mayer and Yung~\cite{DMY} were the first to use the available
processor steps~\cite{KS} as the measure of work for solutions of \DA.
They developed an algorithm which has work $\cO(t+(f+1)p)$ 
and message complexity $\cO((f+1)p)$. 
Galil, Mayer and Yung~\cite{GMY} improved the message complexity to
$\cO(f p^\varepsilon + \min\{f+1, \log p\}p)$,
for any positive $\varepsilon$, while maintaining the same work complexity.
This was achieved as a by-product of their  investigation of the Byzantine agreement with crash failures, for which they found a message-efficient solution.
Chlebus, De Prisco and Shvartsman~\cite{CDS} studied failure models
allowing restarts. 
\remove{
Restarted processors could contribute to the task-oriented work, but the 
cost of integrating them into the system, in terms of the available 
processor steps and communication, might well surpass the benefits.
The solution presented in~\cite{CDS} achieves the work performance
$\cO((t + p\log p + f)\cdot \min\{\log p,\log f\})$, and its message complexity
is $\cO(t+p\log p+ f p)$, against suitably defined adversaries that may
introduce $f$ failures and restarts.
This algorithm is an extension of one that is tolerant of stop-failures 
and which has work complexity 
$\cO((t + p\log p/\log\log p)\log f)$ and communication complexity
$\cO(t + p\log p/\log\log p + fp)$. 
}
Chlebus and Kowalski~\cite{CK} studied the \DA\ problem when occurrences of failures are controlled by the weakly-adaptive linearly-bounded adversary.
They developed a randomized algorithm with the expected effort (\jm{effort = work + number of messages}) $\cO(p\log^*p)$, 
in the case $p=t$, which is asymptotically
smaller than the lower bound $\Omega(p\log p/\log\log p)$ on work 
of any deterministic algorithm.
Chlebus, G\k asieniec,  Kowalski and Shvartsman~\cite{CGKS} developed a
deterministic algorithm with effort $\cO(t+p^a)$, for some
specific constant~$a$, where $1<a<2$, against the
unbounded adversary, which is the first algorithm with the property that both work and communication are $o(t+p^2)$ against this adversary.
They also gave an algorithm achieving both work and communication
$\cO(t+p\log^2 p)$ against a strongly-adaptive linearly-bounded adversary.
All the previously known deterministic algorithms had either work or
communication performance $\Omega(t+p^2)$ when as many as a linear fraction 
of processing units could be failed by a strongly-adaptive adversary. 
Georgiou, Kowalski and Shvartsman~\cite{GKS} developed an algorithm
with work $\cO(t+p^{1+\varepsilon})$, for any fixed
constant~$\varepsilon$, by an approach based on gossiping.
Kowalski and Shvartsman in~\cite{KS03} studied \DA\ in an asynchronous 
message-passing mode when executions are restricted such that 
every message delay is at most~$d$.
They showed lower bound $\Omega(t+pd\log_d p)$ on the expected work. 
They developed several algorithms, among them a deterministic one with 
work $\cO((t+pd)\log p)$. 
For further developments we refer the reader to the book
by Georgiou and Shvartsman~\cite{GSbook}.

\paragraph{Related problems on a shared channel.}

Most of work in this model focused on communication problems, 
see the surveys~\cite{Chl,Gal}. 
Among the most popular protocols for resolving contention on the channel are 
Aloha \cite{Abr} and exponential backoff \cite{MB}.
The two most related research problems are as follows.

The {\em selection problem\/} is about how to have an input message broadcast 
successfully if only some among the stations hold input messages while the other do not. 
\jm{It is somehow closely related to the leader election problem.}
Willard~\cite{Wil} developed protocols solving this problem in the expected time
$\cO(\log\log n)$ in the channel with collision detection.
Kushilevitz and Mansour~\cite{KM} showed a lower bound $\Omega(\log n)$
for this problem in case of a lack of collision detection, which explains
the exponential gap between this model and the one with collision detection.
Martel~\cite{Mar} studied the related problem of finding maximum within the 
values stored by a group of stations.

Jurdzi\'nski, Kuty\l owski and Zatopia\'nski~\cite{JKZ} considered, the leader election problem for the channel without collision detection,
giving a deterministic algorithm with sub-logarithmic energy cost. They also proved \jm{log}-logarithmic lower bound for the problem. 

The {\em contention resolution\/} problem is about a subset of some~$k$ among all $n$~stations which have messages. All these messages need to be transmitted successfully on the channel as quickly as possible.
Koml\'{o}s and Greenberg~\cite{KG} proposed a deterministic solution allowing to achieve this in time $\cO(k+k\log(n/k))$, where ~$n$ and~$k$ are known.
Kowalski~\cite{Kow} gave an explicit solution of complexity~$\cO(k\text{ polylog } n)$, while the lower bound $\Omega(k(\log n)/(\log k))$ was shown by Greenberg and Winograd~\cite{GW}.
The work by Chlebus, Go\l\k ab and Kowalski~\cite{ChleGK} regarded broadcasting spanning forests on a multiple-access channel, with locally stored edges of an input graph.

Significant part of recent results on the communication model considered in the literature is focused on jamming-resistant protocols 
motivated by applications in single-hop wireless networks. To the best of our knowledge this \jm{line} of research  was initiated 
in  \cite{SHaInit} by  Awerbuch, Richa and Scheideler, wherein authors introduced a model of adversary capable of 
jamming up to $(1-\epsilon)$ of the time steps (slots). The following papers~\cite{SH1a, SH2a} by Richa, Scheideler, Schmid and Zhang
proposed several algorithms that can reinforce the communication even for a very strong, adaptive adversary. 
For the same model Klonowski and  Paj\k{a}k in~\cite{SPAAmk} proposed an optimal leader election protocol, using a different algorithmic approach.  

The similar model of a jamming adversary was considered by  Bender, Fineman,  Gibert and Young in~\cite{SETH1}. The authors consider a modified, robust exponential backoff protocol
that requires $O(\log^2 n + T)$ attempts to the channel if there are at most  $T$ jammed slots. Motivated by saving energy, the authors try to find 
maximal global throughput while reducing device costs expressed by the number of attendants in the channel.

\remove{
A very powerful modification of the backoff protocol was proposed by  Bender, Kopelowitz, Pettie and Young~\cite{STOC16Ben}. The authors presented a protocol offering constant 
throughput and only  $O(\log(\log^{*}n))$ expected attempts before getting access to the channel. One of the consequences of this protocol is an ultra-fast  leader election protocol.
}

Finally, there are several recent results on finding approximations of the network. In~\cite{NaszeSA}  Brandes, Kardas, Klonowski, Paj\k{a}k and Wattenhofer proposed an 
algorithm for the network of $n$ stations that returns  $(1+\varepsilon)$-approximation of $n$ with probability  at least $1-1/f$. 
This procedure takes $O(\log\log n+\log f/\varepsilon^2)$ time slots. This result was also proved to be time-optimal in~\cite{NaszeSA}. 
In~\cite{Binbin13}  Chen, Zhou and Yu demonstrated a size approximation protocol for seemingly different model (namely RFID system) that needs 
 $\Omega(\frac{1}{\epsilon^2 \log{1/\epsilon}}  + \log\log{n})$ slots for  $\epsilon \in [1/\sqrt{n},0.5]$ and negligible probability of failure. In fact, this result 
can be instantly translated into the MAC model.

\subsection{Our results}

\remove{
Results in the seminal paper \cite{CKL} were mainly concentrated around deterministic solutions. Nevertheless the authors introduced an efficient algorithm solving Do-All against
the Linearly-Bounded adversary. Additionally they showed a lower bound giving an idea that randomization may be useful against the Weakly-Adaptive adversary. Consequently
that was the seed-reason for this paper.
}

We introduce a hierarchy of adaptive adversaries
and study their impact on the complexity of performing jobs on a shared channel.
The most important parameter of this hierarchy is
the partial order. It \jm{describes adversarial crashes, hence
we call such adversaries ordered.}
The other parameters are: the number of crashes $f$ 
(we call them size-bounded adversaries) and
a delay $c$ in \jm{the effect of the adversary's decisions}.
We call them $c$-Round-Delayed or $c$-RD.

\jm{Since the adversaries that we introduce originate from partial order relations, then appropriate notions and definitions translate straightforwardly.
The relation of our particular interest while considering partially ordered adversaries is the precedence relation.
Precisely, if some station $ v $ precedes station $ w $ in the partial order of the adversary, then we say that $ v $ and $ w $ are comparable. This also means that
station $ v $ must be crashed by the adversary before station $ w $. 
Consequently a subset of stations where every pair of stations is comparable is called a chain. On the other hand a subset of stations where no two different stations are comparable is called an anti-chain.

It is convenient to think about the partial order of the adversary from a Hasse diagram perspective. The notion of chains and anti-chains seems to be intuitive when graphically presented, e.g., a chain
is a pattern of consecutive crashes that may occur while an anti-chain gives the adversary freedom to crash in any order due to non-comparability of elements/stations.
}

We show that adversaries constrained by an order of short width i.e., with short
maximal anti-chain or 1-RD adversaries have very little power, 
results in performance similar to the one enforced by oblivious adversaries or
linearly-ordered adversaries, cf.,~\cite{CKL}.
More specifically, we develop algorithms \Robal\ and \Gilet,
which achieve work performance close to the absolute
lower bound $\Omega(t + p\sqrt{t})$ against ``narrow-ordered'' and 1-RD
adversaries, respectively.

In case of ordered adversaries restricted by orders of arbitrary width $k\le f$,
we present algorithm \Grubtech\ that guarantees 
work $\mathcal{O}(t + p\sqrt{t} + p\min\left\{\frac{p}{p-f},t,k\right\}\log p)$
against ordered adversaries restricted by orders of width $k$, and show that it is efficient by proving a lower bound for a broad class of partial orders.
This also extends the result for a weakly-adaptive linearly-bounded
adversary from~\cite{CKL} to any number of crashes $f<p$, 
as weakly-adaptive adversary is a special case
of ordered adversary restricted by a single anti-chain.
Our results together with \cite{CKL} prove a separation
between classes of adversaries. The easiest to \jm{play against}, apart of
the oblivious ones, are the following adaptive adversaries: 
1-RD adversaries, ordered adversaries restricted by short-width orders, and
linearly bounded adversaries. 
More demanding are ordered adversaries restricted by order of width $k$,
for larger values of $k$, and $f$-bounded adversaries for $f$ close to $p$.
The most \jm{demanding} are strongly-adaptive adversaries, \jm{as their decisions and the way they act are least restricted.}
See Table~\ref{tab1} for detailed results and comparisons.

The hierarchy of the considered adversaries is illustrated on Figure~\ref{fig:hierarchy}.
It depends on three main factors. Additionally,
we introduce several solutions for the specified settings. Consequently our contribution is a complement to adversarial scenarios presented in literature together with a taxonomy describing
the dependencies between different adversaries.

First of all we have the vertical axis which describes adversary features, that is how restricted his decisions are. We have the Strongly-Adaptive adversary in the origin, 
who may decide on-line which stations will be crashed. \jm{Above the Strongly-Adaptive adversary} is the Weakly-Adaptive adversary, who is slightly weaker and before the algorithm execution 
has to declare the subset of stations that will be prone to crashes. \dk{Next we have the $k$-Chain-Ordered adversary and its more general version 
$k$-Thick-Ordered adversary (Chain-O-Adaptive in Figure~\ref{fig:hierarchy} for consistency). 
Apart from declaring the faulty subset, the former adversary is restricted by partial orders being collections of disjoint $k$ chains, 
while the latter --- by all partial orders of thickness $k$ (and so $k$-chains as well).
Finally, there is the Linearly-Ordered adversary (Linearly-O-Adaptive in Figure~\ref{fig:hierarchy}) that we introduce in this paper --- its order is described by a linear pattern of processor crashes.}

The horizontal axis describes another \jm{feature of the adversary}, that we introduce in our paper, i.e., the Round-Delay of adversary decisions. 
Similarly, the configuration for the problem is hardest in the origin and a 0-RD adversary is the strongest against which we may execute an algorithm. 
An interesting particularity is that if the Strongly-Adaptive adversary's \jm{decisions are delayed by at least one round}, then we may design a solution whose work complexity is independent of the number crashes.

The axis orthogonal to those already considered, describes the channel feedback. In the origin we have a multiple-access channel without collision detection, then there is the beeping channel, followed by
MAC with collision detection.

We may see that the most difficult setting is in the origin, \jm{while the further from the origin, the easier the problem}. The boxes in Figure \ref{fig:hierarchy} represent the algorithms and their work complexities
in certain configurations of the model i.e. features described above.
The bold boxes denote algorithms from this paper and the remaining ones are from CKL \cite{CKL}. Factors marked red denote the ``distance'' from the lower bounds, understood as 
how far the algorithms are from optimum.

\begin{figure}[htb]
 \begin{center}
 \includegraphics[scale=0.485]{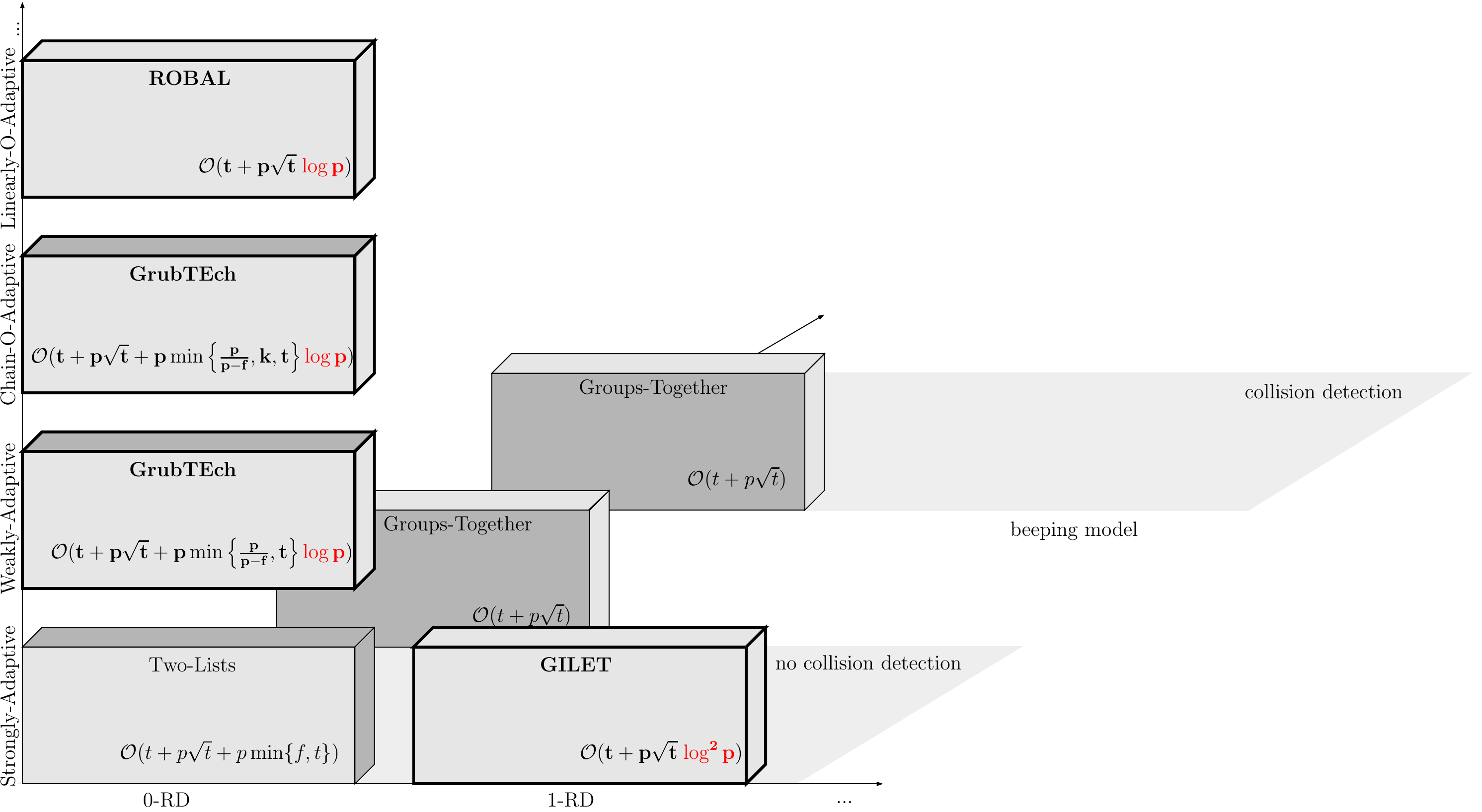}
 \caption{The hierarchy of adversaries.}
 \label{fig:hierarchy}
 \end{center}
\end{figure}

\begin{table}[htbp]
  \scalebox{0.92}{
  \begin{tabular}{l||l|l|l|l|l|l}

    \scriptsize{\textbf{algorithm}} &  \scriptsize{\textbf{channel}} & \scriptsize{\textbf{adv.}} & \scriptsize{\textbf{work}} & \scriptsize{\textbf{ref.}} & \scriptsize{\textbf{lower bound}} & \scriptsize{\textbf{ref.}} \\ \hline\hline
    \scriptsize{Two-Lists} & \scriptsize{no-CD} & \scriptsize{SA} & \tiny{$ \mathcal{O}(t + p\sqrt{t} + p\min\{f, t\}) $} & \scriptsize{\cite{CKL} Thm 1} & \tiny{$ \Omega(t + p\sqrt{t} + p\min\{f, t\}) $} & \scriptsize{\cite{CKL} Thm 2} \\ \hline
    \scriptsize{Groups-Together} & \scriptsize{CD} & \scriptsize{SA} & \tiny{$ \mathcal{O}(t + p\sqrt{t}) $} & \scriptsize{\cite{CKL} Thm 3} & \tiny{$ \Omega(t + p\sqrt{t}) $} & \scriptsize{\cite{CKL} Lem 2} \\ \hline
    \scriptsize{Mix-Rand} & \scriptsize{no-CD} & \scriptsize{WALB} & \tiny{$ \mathcal{O}(t + p\sqrt{t}) $} & \scriptsize{\cite{CKL} Thm 5} & \tiny{$ \Omega(t + p\sqrt{t}) $} & \scriptsize{\cite{CKL} Lem 2} \\ \hline
    \scriptsize{ROBAL} & \scriptsize{no-CD} & \scriptsize{LOA} & \tiny{$ \mathcal{O}(t + p\sqrt{t}\log p) $} & \scriptsize{Sec. \ref{robal}} & \tiny{$ \Omega(t + p\sqrt{t}) $} & \scriptsize{\cite{CKL} Lem 2} \\ \hline
    \scriptsize{GrubTEch} & \scriptsize{no-CD} & \scriptsize{WA} & \tiny{$ \mathcal{O}(t + p\sqrt{t} + p\min\left\{\frac{p}{p-f},t\right\}\log p) $} & \scriptsize{Sec. \ref{grubtech}} & \tiny{$ \Omega(t + p\sqrt{t} + p\min\left\{\frac{p}{p-f},t\right\} ) $} & \scriptsize{\cite{CKL} Thm 6} \\ \hline
    \scriptsize{GrubTEch} & \scriptsize{no-CD} & \scriptsize{COA} & \tiny{$ \mathcal{O}(t + p\sqrt{t} + p\min\left\{\frac{p}{p-f},t,k\right\}\log p) $} & \scriptsize{Sec. \ref{grubtechpartord}} & \tiny{$ \Omega(t + p\sqrt{t} + p\min\left\{\frac{p}{p-f},k,t\right\} ) $} & \scriptsize{Sec. \ref{grubtechpartord}} \\ \hline
    \scriptsize{GILET} & \scriptsize{no-CD} & \scriptsize{1-RD} & \tiny{$ \mathcal{O}(t + p\sqrt{t}\log^{2}p) $} & \scriptsize{Sec. \ref{gilet}} & \tiny{$ \Omega(t + p\sqrt{t}) $} & \scriptsize{\cite{CKL} Lem 2} \\ \hline
  \end{tabular}
  }
  \caption{Summary of main results; first three were introduced in CKL \cite{CKL}, the other are presented in this paper. CD stands for collision detection model feature.
	SA stands for Strongly-Adaptive adversary, WA (WALB) stands for Weakly-Adaptive
	(Linearly-Bounded) adversary, COA stands for Chain-Ordered adversary,
	LOA stands for Linearly-Ordered adversary, and 1-RD stands for 1-Round-Delay adversary.}
  \label{tab1}
  \end{table}


A subset of the results from this paper forms a hierarchy of partially ordered adversaries. The first solution we introduced was designed to work against a linearly-ordered adversary, whose pattern
of crashes is described by a linear order. The upper bound of this algorithm does not depend on the number of crashes and is just logarithmically far from the minimal work complexity in the assumed model.
The second algorithm serves for the case when the adversary's partial order of stations forms a maximum length anti-chain. Nevertheless we also analyze this solution against an in-between situation
when the partial order is consists of $ k $ chains of an arbitrary length, yet the sum of their lengths is $ f $.

In order to conclude the content of this paper, we would like to emphasize that building on solutions from CKL \cite{CKL} we introduce different algorithms and specific adversarial scenarios,
for more complex setups that, to some extent, filled the gaps for randomized algorithms solving Do-All in the most challenging adversarial scenarios and communication channels providing least feedback.
Due to the basic nature of the considered communication model ---
a shared channel with acknowledgments only --- our solutions
are also implementable and efficient in various different types of communication 
models with contention and failures.

\subsection{Document structure}

We describe the model of the problem, communication channel details, different adversary scenarios and the complexity measure in Section \ref{model}. Section \ref{algtools} is dedicated 
to the \textsc{Two-Lists} and \textsc{Groups-Together} procedures from \cite{CKL} that ares used 
(sometimes after small modifications) as a toolbox in
our solutions.
In Section \ref{robal} we present a randomized algorithm \Robal\ solving \DA\ in presence 
of a Linearly-Ordered adversary. In the following section (Section \ref{grubtech}) there is a work-efficient algorithm \Grubtech\ that simulates a kind of fault-tolerant collision detection on a channel without such feature. 
This is followed by Section \ref{grubtechpartord}, where we adjust this solution for a k-Chain-Ordered adversary. Finally, we have Section \ref{gilet} that contains a solution for the $1$-RD adversary (algorithm \Gilet) and Section \ref{beeping} is dedicated to the transition of~\GT~to the beeping model. \mj{We conclude with a short summary in Section \ref{conclusions}, which also contains some general remarks about time and energy complexity of our algorithms. }

%% file: model.tex
\section{Shared channel model and the Do-All problem}
\label{model}

The Do-All problem has been introduced by Dwork et al \cite{DHW} and was considered further in numerous papers \cite{CPS, CGKS, CK, DMY, GMY} under different assumptions
regarding the model. In this section we formulate the model we consider, which is based on the one considered in \cite{CKL}.

In general, the Do-All problem is modeled as follows: a distributed system of computationally constrained devices is expected to perform a number of tasks \cite{GSbook2}.
We will call those devices \textit{processors} or simply \textit{stations}. The main efficiency measure that we use is \textit{work}, i.e., the total number of processor steps available
for computations \cite{KS}.

\subsection{Stations}

In our model we assume having $ p $ stations, with unique identifiers from the set $ \{1,\dots, p\} $. The distributed system of those devices is synchronized with a global clock,
and time is divided into synchronous time slots, called \textit{rounds}. All the stations start simultaneously at a certain moment. Furthermore every station may halt voluntarily.
\remove{
A halted station 
does not contribute to the work done, in contrast to an operational station,
but contributes to the work performance measure. 
Precisely, every operational station generates a unit of work per round, that has to be included while computing
an algorithm's complexity.
}
In this paper, by $ n < p $ we will denote the number of operational, i.e. not crashed, stations.

\subsection{Communication}

The communication channel for processors is a multiple-access channel \cite{Chl, Gal}, where a transmitted message reaches every operational device.
All our solutions work on a channel \textit{without collision detection}, hence
when more than one message is transmitted at a certain round, then the devices hear a signal indistinguishable from the background noise.
In our model we assume, that the number of bits possible to broadcast in a single transmission is bounded by $ \mathcal{O}(\log p) $, however all our algorithms broadcast merely $ \mathcal{O}(1) $ bits,
hence we omit the analysis of the message complexity.

\dk{In few places of this work we refer to the alternative channel setting
\textit{with collision detection}, in which there are three types of feedback from 
the channel:
\begin{itemize}
 \item \textbf{Silence} --- no station transmits, and only a background noise is heard;
 \item \textbf{Single} --- exactly one station transmits a legible information;
 \item \textbf{Collision} --- an illegible signal is heard (yet different from \textit{Silence}), when more than one station transmits simultaneously.
\end{itemize}
Section~\ref{grotog} provides more details.
We note here that the setting with collision detection is referred to only
in the context of transforming algorithmic tools and lower bounds
to the more challenging setting without collision detection primarily studied
in this work.}

It is worth emphasizing that the communication channel in our model
can be made resistant to non-synchronized processor clocks
\dk{without increasing asymptotic performance,} using methods developed previously \cite{HSbook}:
\dk{when processor clocks are not synchronized, then we could replace each round 
by two rounds of appropriate lengths to compensate possible lags.}

\subsection{Adversaries} \label{adversaries}

Processors may be crashed by the adversary. One of the factors that describes the adversary is its power $ f $. It represents the total number of failures that 
may be 
enforced.
We assume that $ 0 \leq f \leq p-1 $, so always at least one station remains operational until an algorithm terminates. Stations that were crashed neither restart nor contribute to work.
\jm{Another feature of the adversary is whether it is adaptive or not. Following the definition from \cite{GSbook2}, an \textit{adaptive adversary} is the one that has complete knowledge of the computation that it is
affecting, and it makes instant dynamic decisions on how to affect the computation. A 
\textit{non-adaptive adversary}, sometimes called \textit{oblivious}, has to determine a sequence of events it will cause before the start of the computation.
In this paper we focus on adaptive adversaries.}

\dk{In the previous work, the following adversarial models were considered, c.f.,~\cite{CKL,GSbook}:}

\noindent \textbullet \: \textbf{Strongly-Adaptive $f$-Bounded}: the only restriction of this adversary is that the total number of failures may not exceed $ f $. In particular all possible failures may happen
 simultaneously.
\newline \textbullet \: \textbf{Weakly-Adaptive $f$-Bounded}: the adversary has to declare a subset of $ f $ stations prone to crashes before the algorithm execution. \jm{Yet, it may arbitrarily perform crashes on
the declared subset.}
\newline \textbullet \: \textbf{Unbounded}: that is Strongly-Adaptive $ (p-1) $-Bounded.
\newline \textbullet \: \textbf{Linearly-Bounded}: an adversary of power $ f $, where $ f = cp $, for some $ 0 < c < 1 $.

We \dk{now} introduce new \dk{adversarial models}
that complement the existing 
\dk{ones from the literature.} 

\subsubsection{The Ordered $f$-Bounded adversary}

\dk{ Formally, the \textit{Ordered $f$-Bounded adversary} has to declare, prior the execution, a subset of at most $ f $ out of $ p $ stations that will be prone to crashes. Then, before starting the execution, the adversary has 
to impose a partial order on the selected stations, taken from a given
family of partial orders. 
This family restricts the power of the adversary ---
the wider it is the more power the adversary possesses. Moreover,
as we will show in this work, the structure of available partial orders 
influences asymptotic performance of algorithms and the complexity
of the Do-All problem under the presence of the adversary restricted
by these partial orders. 

The adversary may enforce a failure independently from time slots (even $ f $ at the same round),
but with respect to the order. 
This means that a pre-selected crash-prone station can be crashed in a time slot
if and only if 
all stations preceding it in the order has been already crashed by the end of 
the time slot.

In this work we focus on the following three types of partial orders.} 

\paragraph{The Linearly-Ordered $f$-Bounded adversary.}

Formally, the \textit{Linearly-Ordered $f$-Bounded} 
has to choose a \jm{sequence
$ \pi = \pi(1)\dots\pi(f) $} designating the order on the selected set of $f$
stations in which the failures will occur, \jm{where $\pi(i) $ represents the id of the $i$th fault-prone station in the order.}
This means that station $ \pi(i) $ may be crashed if and only if stations $ \pi(j) $ are already crashed, for all $ j < i $. 
\jm{In what follows, the notion of sequence $ \pi $ is consistent with a linear partial order.}

\paragraph{The $k$-Chain-Ordered and $k$-Thick-Ordered $f$-Bounded adversary.}

The \textit{$k$-Chain-Ordered $f$-Bounded} adversary 
\dk{has to arrange the pre-selected $ f $ stations into} a partial order consisting of $ k $ \jm{disjoint} chains of arbitrary length that represent in what order these stations may be crashed.
In what follows there are $ k $ chains; 
we denote $ l_{j} $ as the length of chain $ j $, and we assume that the sum of lengths
of all chains equals $ f $.
While considering $k$-Chain-Ordered adversaries we will also define additional notions, useful in the analysis of certain results.
We say that a partial order is a {\em $k$-chain-based partial order}
if it consists of $k$ disjoint chains such that: 
\begin{itemize}
\item
no two of them have a common successor, and 
\item
the total length of the chains is a constant fraction of all elements in the order.
\end{itemize}
Furthermore, by the \textit{thickness} of a partial order $P$ we understand the maximal size of an anti-chain in~$P$.
\dk{An adversary restricted by a wider class of partial orders of thickness $k$
is called \textit{$k$-Thick-Ordered}.}

\paragraph{The Anti-Chain-Ordered $f$-Bounded adversary.}
\dk{This adversary is restricted by a partial order which is the anti-chain of
$f$ elements, i.e., all $f$ crash-prone stations are incomparable, thus could
be crashed in any order.
This adversary is the same as the Weakly-Adaptive $f$-Bounded adversary and
the $f$-Thick-Ordered ones.}

\subsubsection{The c-RD f-Bounded adversary}

The \textit{c-RD} adversary decisions take effect with a $ c $ round delay. This means that if we consider time divided into slots (rounds), then if the adversary decides to interfere with
the system (crash a processor) then this will \dk{inevitably} happen after $ c $ rounds. \dk{In particular, this means that the subsequent execution and random bits do not influence the decision and its effect ---
the decision is final and once made by the adversary cannot be changed during the delay.} 
We still consider $f$-boundedness of the adversary, but apart from that 
it may decide arbitrarily, without declaring which stations will be prone to crashes before the algorithm's execution.

A special case of the $c$-RD adversary is a \textit{$0$-RD} and a \textit{$1$-RD} adversary model. The definition of the former case is consistent with the Strongly-Adaptive adversary.
The latter case may give an answer to the question regarding the matter of how delay influences the difficulty of the problem for a strong adversary.

\subsection{Complexity measure}
\label{subsec-work}

The complexity measure that is mainly used in our analysis is \textit{work}, as mentioned before. It is the number of available processor steps for computations. This means that each operational station
that did not halt contributes a unit of work even if it is idling. \mj{Since we use work complexity measure extensively, we adopt the following definition from \cite{GSbook}.}

\mj{
\begin{definition} \label{workdef} (\cite{GSbook}, Definition 2.2)
Let $ A $ be an algorithm that solves a problem of size $ t $ with $ p $ processors, under adversary $ \mathcal{A} $. \jmii{Let $ \mathcal{E}(A, \mathcal{A}) $ 
denote the set of all executions of algorithm $ A $, under adversary $\mathcal{A}$.} For execution $ \xi \in \mathcal{E}(A, \mathcal{A}) $, let time $ \tau(\xi)$ be the
time (according to the external clock), by which A solves the problem. By $p_{i}(\xi)$ let us denote the number of processors completing a local computation step (e.g., a machine instruction)
at time $i$ of the execution, according to some external global clock (not available to the processors) . Then \textit{work complexity} $ S $ of algorithm $ A $ is:
$$ S = S_{\mathcal{A}}(t, p) = \max_{\xi \in \mathcal{E}(A, \mathcal{A})}\left\{ \sum_{i = 1}^{\tau(\xi)} p_{i}(\xi) \right\}. $$
\end{definition}
}

\mj{For randomized algorithms that we are dealing with in this paper, we assess the \textit{expected work} $ S^{E}_{\mathcal{A}}(t, p) $, which is defined as the maximum over all executions $ \xi \in \mathcal{E}(A, \mathcal{A}) $ of the expectation of the sum 
	$\sum_{i = 1}^{\tau(\xi)} p_{i}(\xi)$ from Definition \ref{workdef}.}

\mj{In order to illustrate work complexity measure of a single execution of an algorithm} let us assume that an execution starts when all the stations
begin simultaneously in some fixed round $ r_{0} $. Let $ r_{v} $ be the round when station $ v $ halts or is crashed. Then its work contribution is equal $ r_{v} - r_{0} $. 
In what follows, the algorithm complexity is the sum of such expressions over all stations, i.e.: $ \sum_{0 \leq v \leq p}(r_{v} - r_{0}) $.

\remove{
\noindent 
\mj{At the end of this paper in Section \ref{timeenergy} we also elaborate on time and (transmission) energy complexity measures of the designed algorithms.
Time complexity for Do-All protocols on a shared channel is a notion adopted from \cite{CMS}, while the energy complexity reflects a similar measure used for the Do-All problem in the message-passing model,
i.e., the number of messages sent, despite their bit size in the entire computation.
We define them as follows.}

\begin{definition}
 Time complexity of an algorithm execution is the number of steps until all the processors either have halted or have failed.
\end{definition}

\begin{definition}
 Energy complexity of an algorithm execution is the number of transmissions until all the processors either have halted or have failed.
\end{definition}

}

\subsection{Tasks and reliability}

We expect that processors will perform all $ t $ tasks as a result of executing an algorithm. \jm{Tasks are initially known to processors.}
We assume that tasks are \textit{similar} (that is each task requires the same number of rounds to be done),
\textit{independent} (they can be performed in any order) and \textit{idempotent} (every task may be performed many times, even concurrently by different processors without affecting the outcome of its computation).
We assume that one round is sufficient to perform a single task.

\subsection{Do-All formal definition}
Having explained the assumptions for tasks, we may now state the formal definition of the Do-All problem after \cite{GSbook2}:

\medskip
\mj{\textit{Do-All: Given a set of $ t $ tasks, perform all tasks using $ p $ processors, under adversary $\mathcal{A}$.}}

\medskip
\noindent
\mj{In our considerations adversary $ \mathcal{A} $ from the definition above is one of the adversaries described in Section \ref{adversaries}}.

We assume that all our algorithms need to be reliable. A \textit{reliable} algorithm satisfies the following conditions in any execution:
\textit{all the tasks are eventually performed, if at least one station remains non-faulty} \textbf{and} \textit{each station eventually halts, unless it has crashed}.

\jm{The Do-All problem may be considered completed or solved.
It is considered {\em completed} when all tasks are performed, but their outcomes are not necessarily
known by all operational stations. The problem is considered {\em solved} if in addition all operational processors are aware of the tasks' outcomes.
In this paper we do not assume that stations need to know tasks' outcomes, yet algorithm \Robal\ is designed in such a way that it may solve the problem, while all other algorithms complete it (performing tasks
is confirmed by a collision signal that does not contain any meaningful message).}

%% file: algtools.tex
\section{Useful algorithmic tools}
\label{algtools}

\subsection{Algorithm Two-Lists}

In this subsection we describe a deterministic \textit{Two-Lists} algorithm from \cite{CKL} which is used in our solutions as a sub-procedure.
It was proved that this algorithm is asymptotically optimal for the Weakly-Adaptive adversary on a channel without collision detection, and its work complexity is 
$ \mathcal{O}(t + p\sqrt{t} + p\:\min\{f, t\}) $. The characteristic feature of Two-Lists is that its complexity is linear for some setups of $ p$ and $ t $ parameters,
describing the number of processors and tasks, respectively (for details, see Fact \ref{fact13}).

\begin{algorithm}
{
{set pointer $ \texttt{Task\_To\_Do}_{v} $ on list $ \texttt{TASKS} $ to the initial position of the range $ v $\;}
{set pointer $ \texttt{Transmit} $ to the first item on list $ \texttt{STATIONS} $\;}
{\Repeat{(pointer $ \texttt{Transmit} $ points to the first entry on list $\texttt{STATIONS} $) or (all tasks in list $ \texttt{TASKS} $ have been covered in the epoch)}
{

{\tcp{Round 1:}}
{perform the first task on list $ \texttt{TASKS} $, starting from the one pointed to by $ \texttt{Task\_To\_Do}_{v} $, that is in list $\texttt{OUTSTANDING}_{v} $}
{move the performed task from list $\texttt{OUTSTANDING}_{v} $ to list $\texttt{DONE}_{v} $\;}
{advance pointer $ \texttt{Task\_To\_Do}_{v} $ by one position on list $ \texttt{TASKS} $\;}
{\tcp{Round 2:}}
{\If{$ \texttt{Transmit} $ points to $ v $}{broadcast one bit\;} }
{attempt to receive a message\;}
{\tcp{Round 3:}}
{\If{a broadcast was heard in the preceding round}{
{\For{each item $ x $ on list $\texttt{DONE}_{Transmit} $}{
{\If{$ x $ is on list $ \texttt{OUTSTANDING}_{v} $}{move $ x $ from $\texttt{OUTSTANDING}_{v} $ to $\texttt{DONE}_{v} $\;}}
{\If{$ x $ is on list $ \texttt{TASKS} $}{remove $ x $ from $ \texttt{TASKS} $\;}}
}}
{\If{list $ \texttt{TASKS} $ is empty}{halt\;}}
{advance pointer $ \texttt{Transmit} $ by one position on list $ \texttt{STATIONS} $\;}
}}\Else{remove the station pointed to by $ \texttt{Transmit} $ from $ \texttt{STATIONS} $\;}

}}}
\caption{\textsc{Epoch-Two-Lists}, code for station $ v $; from \cite{CKL}}
\label{algorytm01}
\end{algorithm}

\begin{algorithm}
{
{- initialize $\texttt{STATIONS}$ to a sorted list of all $p$ names of stations\;}
{- initialize both $\texttt{TASKS}$ and $\texttt{OUTSTANDING}_{v} $ to sorted list of all $ t $ names of tasks\;}
{- initialize $\texttt{DONE}_{v} $ to an empty list of tasks\;}
{- \Repeat{halted}{\textsc{Epoch-Two-Lists}\;}}
}
\caption{\textsc{Two-Lists}, code for station $ v $; from \cite{CKL}}
\label{algorytm02}
\end{algorithm}

\subsubsection{Basic facts and notation}

\textsc{Two-Lists} was designed for a channel without collision detection. That is why simultaneous transmissions were excluded therein. It has been realized by a cyclic schedule of broadcasts
(round-robin). This means that stations maintain a transmission schedule and broadcast one by one, accordingly. Because of such design every message transmitted via the channel
is legible for all operational stations.

Another important fact about \textsc{Two-Lists} is that stations maintain the list of tasks, what enables them to distinguish which tasks are they responsible for. Both the tasks list and the transmission
schedule are maintained as common knowledge. The result of such an approach is that stations may transmit messages of a minimal length, just to confirm that they are still operational
and performed their assigned tasks.

Additionally the transmission schedule and tasks list is stored locally on each station, but the way how stations communicate allows to think of those lists as common for all operational stations.

\textsc{Two-Lists} is structured as a loop (see Algorithm \ref{algorytm01}). Each iteration of the loop is called an \textit{epoch}. Every \textit{epoch} begins with a transmission schedule and tasks being
assigned to processors. During the execution some tasks are performed and if a station transmits such fact, it is confirmed by removing those certain tasks from list \texttt{TASKS}.
However due to adversary activity some stations may be crashed, what is recognized as silence heard on the channel in a round that a station was scheduled to transmit. Stations recognized
as crashed are also removed from the transmission schedule. Eventually a new \textit{epoch} begins with updated lists.

\textit{Epochs} are also structured as loops (see Algorithm \ref{algorytm02}). Each iteration is now called a \textit{phase}, that consists of three consecutive rounds in which station $ v $:
 \begin{enumerate}
  \item Performs the first unaccomplished task that was assigned to $ v $;
  \item $ v $ broadcasts one bit, confirming the performance of tasks that were assigned to $ v $, if it was $ v $'s turn to broadcast. Otherwise $ v $ listens to the channel and attempts to receive
  a message
  \item Depending on whether a message was heard $ v $ updates its information about stations and tasks.
 \end{enumerate}

An epoch consists of a number of phases, that is described by the actual number of operational stations or outstanding tasks.
In each epoch there is a repeating pattern of phases that consists of the following three rounds: (1) each operational station performs one task. Next (2) a transmission round takes place,
where at most one station broadcasts a message, and the rest of the stations attempt to receive it. The process is ended (3) by an updating round, where stations reconstruct their knowledge about
operational stations and outstanding tasks.

\subsubsection{The significance of lists}

In the previous section we mentioned the concept of \textit{knowledge} about stations and tasks, that processors maintain. It was described somehow abstractly, so now we will explain it in detail.
Furthermore we will provide information on how the stations are scheduled to transmit and how do they know which tasks should they perform.

It is not accidental that the algorithm was named \textsc{Two-Lists} as the most important pieces of information about the system are actually maintained on two lists. The first is list \texttt{STATIONS}.
It represents operational (at the beginning of an epoch) processors and sets the order in which stations should transmit in consecutive phases. That list is operated by pointer
\texttt{Transmit}, that is incremented after every phase. It points exactly one station in a single iteration, what prevents collisions on the channel. Hence when some station did not
broadcast we may recognize that it was crashed and eliminate from \texttt{STATIONS}, setting the pointer to the following device.

\begin{figure}[htb]
 \begin{center}
 \includegraphics[scale=0.7]{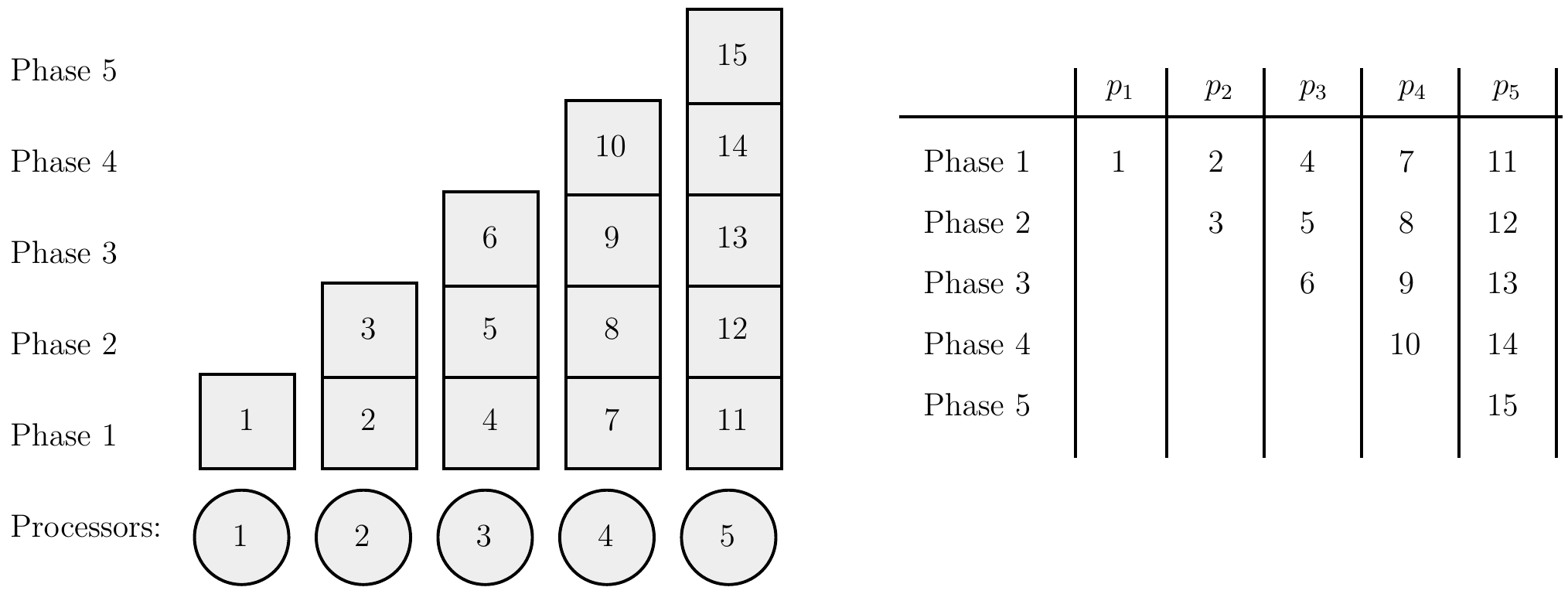}
 \caption{Tasks assignment in \textsc{Two-Lists}.}
 \label{fig:assignment}
 \end{center}
\end{figure}

The second list is \texttt{TASKS}. It contains outstanding tasks, and the associated pointer is \\\texttt{Task\_To\_Do$_{v}$}, separate for each station. Task assignment is organized in the following way
\jm{(see Figure \ref{fig:assignment} for a visualized example).}
Let us present processors from list \texttt{STATIONS} as a sequence $\langle v_{i} \rangle_{1 \leq i \leq n}$, where $ n = \texttt{|STATIONS|} $ is the number of operational stations at the beginning of the epoch.
Each station is responsible for some segment of list \texttt{TASKS} and all segments sum to the whole list. The length of a segment for station $ v_{i} $ equals $ i $ in a single epoch.
A single task may belong to more than
one segment at a time, unless the number of tasks is accordingly greater than the number of stations.

It is noticeable that lists \texttt{STATIONS} and \texttt{TASKS} are treated as common to all the devices, because of maintaining common knowledge.
However, in fact every station has a private copy of those lists and operates with appropriate pointers.

Finally, there are additional two lists maintained by each station. The first one is list $\texttt{OUTSTANDING}_{v} $ and it contains the segment of tasks that station $ v $ has assigned to perform in an epoch.
The second is list $\texttt{DONE}_{v} $ and it contains tasks already performed by station $ v $. These two additional lists are auxiliary and their main purpose is to structure algorithms in a clear and readable way.

\subsubsection{Sparse vs dense epochs}

The last important element of \textsc{Two-Lists} description, that explains some subtleties are definitions of \textit{dense} and \textit{sparse} epochs.

\begin{definition}
 Let $ n = \texttt{|STATIONS|} $ denote the number of operational stations at the beginning of the epoch. If  $ n(n + 1)/2 \geq \texttt{|TASKS|} $ then we say that an epoch is dense.
 Otherwise we say that an epoch is sparse.
 \label{def01}
\end{definition}

The expression $ n(n + 1)/2 = 1 + 2 + \cdots + n $ from the definition above determines how many tasks may be performed in a single epoch. If all the broadcasts in \textsc{Two-Lists} are successful,
then this is the number of performed (and confirmed) tasks.

In general that is why if we consider a dense epoch, then it is possible that some task $ i $ was assigned more than once to different stations. A dense epoch may end when the list of tasks will become empty.
However for sparse epochs the ending condition is consistent with the fact that every station had a possibility to transmit, and pointer \texttt{Transmit} passed all the devices on list \texttt{STATIONS}.

We end this section with results from \cite{CKL} stating that \textsc{Two-Lists} is asymptotically work optimal, for the channel without collision detection and against the Strongly-Adaptive adversary.

\begin{fact} \label{theorem01}
 (\cite{CKL}, Theorem 1) Algorithm Two-Lists solves Do-All with work $\mathcal{O}(t + p\sqrt{t} + p\min\{f, t\}) $ against the f-Bounded adversary, for any $ 0 \leq f < p $.
\end{fact}

\begin{fact}
 (\cite{CKL}, Theorem 2) The f-Bounded adversary, for $ 0 \leq f < p $, can force any reliable, possibly randomized, algorithm for the channel without collision detection to perform work 
 $ \Omega(t + p\sqrt{t} + p\min\{f, t\}) $.
\end{fact}

\begin{fact}
 (\cite{CKL}, Corollary 1) Algorithm Two-Lists is optimal in asymptotic work efficiency, among randomized reliable algorithms for the channel without collision detection, against the adaptive adversary
 who may crash all but one station.
\end{fact}

\subsection{Algorithm Groups-Together} \label{grotog}
\jm{
Beside \textsc{Two-Lists}, which serves as a sub-procedure for our considerations, some of our algorithms are built on another algorithm from CKL \cite{CKL} --- \textsc{Groups-Together}.
The design of both algorithms is similar, yet \textsc{Groups-Together} was introduced for a channel with collision detection. In what follows, we will describe the technicalities and main differences in this 
subsection.

\noindent Let us recall that a shared channel with collision detection provides three types of signals:

\begin{itemize}
 \item \textbf{Silence} --- no station transmits, and only a background noise is heard;
 \item \textbf{Single} --- exactly one station transmits a legible information;
 \item \textbf{Collision} --- an illegible signal is heard (yet different from \textit{Silence}), when more than one station transmits simultaneously.
\end{itemize}

Simultaneous transmissions are excluded in \textsc{Two-Lists}, as they do not provide any valuable information when collision detection is not available.
In such case a simultaneous transmission of multiple stations results in a silence signal,
\dk{and does not provide any meaningful information.}

Because \textsc{Groups-Together} is specifically designed to work on a channel with collision detection, \dk{then the feedback from collision signals is extensively used.} 
The main difference is that instead of list \texttt{STATIONS}, it maintains list \texttt{GROUPS} --- and indeed, in \textsc{Groups-Together} the stations are arranged into \dk{disjoint} groups. Assigning stations to groups is as follows.
Let $ n $ be the smallest number such that $ n(n+1)/2 > |\texttt{TASKS}| $ holds.
Stations have their unique identifiers from set $ \{1,\dots, p\} $. Let $ g_{i} $ denote some group $ i $, where $ g_{i} $ contains the stations that identifiers are congruent modulo $ i $. For this reason, any two
groups from \texttt{GROUPS} differ in size by at most 1. Consequently, the initial partition results in having $ \min\{\sqrt{t}, p\} $ groups.

Tasks assignment is the same as in \textsc{Two-Lists}, with the difference that now the algorithm operates on groups instead of single stations.
In what follows, all the stations within a single group have the same tasks assigned and hence work together on exactly the same tasks. 
The round-robin schedule of consecutive broadcasts from \textsc{Two-Lists} also applies to \textsc{Groups-Together}, yet now points out particular groups 
instead of single stations. Consequently, if a group broadcasts simultaneously and there is a collision signal (or \dk{a single transmission}) heard on the channel, this means that the tasks
that the group was responsible for have been actually performed and may be removed from list \texttt{TASKS}. However, if silence is heard, then we are sure that all the stations from the group have been crashed.

Apart from the differences described above \textsc{Groups-Together} is the same \textsc{Two-Lists}. It is structured as a loop, which one iteration is called an epoch. An epoch is also structured as a repeat
loop which one iteration is called a phase. Phases contain three rounds, one of which is for transmission. If no transmission occurs in a phase, we call it \textit{silent}. Otherwise it is called \textit{noisy}.
The notions of dense and sparse epochs remain the same as in the \textsc{Two-Lists} analysis.

We finish this section with useful results from CKL~\cite{CKL} about \textsc{Groups-Together}.

\begin{fact}
 (\cite{CKL}, Lemma 4) Algorithm \textsc{Groups-Together} is reliable.
 \label{lemma21}
\end{fact}

\begin{fact}
 (\cite{CKL}, Theorem 3) Algorithm \textsc{Groups-Together} solves Do-All with the minimal work $ \mathcal{O}(t + p\sqrt{t}) $ against the $ f $-Bounded adversary, for any $ f $ such that $ 0 \leq f < p $.
 \label{theorem21}
 \end{fact}
}

%% file: robal.tex
\section{ROBAL --- Random Order Balanced Allocation Lists}
\label{robal}

In this section we describe and analyze the algorithm 
for the \DA\
problem in the presence of a Linearly-Ordered adversary
on a channel without collision-detection. Its expected work complexity is $ \mathcal{O}(t + p\sqrt{t}\log(p)) $ and it uses the \textsc{Two-Lists} procedure from \cite{CKL}
(c.f., Section~\ref{algtools}).

\begin{algorithm}
\If{$ p^{2} \leq t $}{execute \textsc{Two-Lists}\;}
\Else{
{initialize $ \texttt{STATIONS} $ to a sorted list of all $ p $ names of stations\;}
\If{$\log_{2}(p) > e^{\frac{\sqrt{t}}{32}}$}
{execute a $t$-phase epoch where every station has all tasks assigned (without transmissions)\;
execute \textsc{Confirm-Work}\;}
\Else{

{$ i = 0 $\;}
\Repeat{$ i = \lceil \log_{2}(p) \rceil $}{

\If{($ \frac{p}{2^{i}} \leq \sqrt{t} $)}
{Execute \textsc{Two-Lists}\;}
\If{\textsc{Mix-And-Test($ i, t, p $)}}{
\Repeat{less than $ \frac{1}{4}\sqrt{t} $ broadcasts are heard}{Execute $ \sqrt{t} $ phases of \textsc{Two-Lists}\;}
}
{increment $ i $ by $ 1 $\;}
}
}
}
\caption{\textsc{ROBAL}, code for station $ v $}
\label{algorithm12}
\end{algorithm}

\begin{algorithm}
{$i := 0$ \;}
\Repeat{a broadcast was heard}{

{$\texttt{coin} := \frac{p}{2^{i}}$ \;}
{toss a coin with the probability $ \texttt{coin}^{-1} $ of heads to come up\;}
\If{heads came up in the previous step}{broadcast $ v $ via the channel and attempt to receive a message\;}
\If{some station $ w $ was heard}{
clear list $\texttt{TASKS}$\;
break\;
}
\Else{
  increment $ i $ by $ 1 $ \;
  \If{ $ i = \lceil \log_{2}(p) \rceil + 1 $ }
  {$ i := 0 $\;}
}
}

\caption{\textsc{Confirm-Work}, code for station $ v $}
\label{algorithm13}
\end{algorithm}

\begin{algorithm}
\KwIn{$i, t, p$}
{$\texttt{coin} := \frac{p}{2^{i}}$\;}
{\For{$\sqrt{t}\log(p)$ times}
{
{\If{$v$ has not been moved to front of list $ \texttt{STATIONS} $ yet}{toss a coin with the probability $ \texttt{coin}^{-1} $ of heads to come up}}
{\If{heads came up in the previous step}{broadcast $ v $ via the channel and attempt to receive a message}}
{\If{some station $ w $ was heard}{
{move station $ w $ to the front of list $ \texttt{STATIONS} $\;}
{decrement $ \texttt{coin}$ by $ 1 $\;}
}}
}}
{\If{at least $\sqrt{t} $ broadcasts were heard}{return \texttt{true}\;}
\Else{return \texttt{false}\;}
}

\caption{\textsc{Mix-And-Test}, code for station $ v $ }
\label{algorithm11}
\end{algorithm}

\jm{\textsc{ROBAL} \mj{(Algorithm \ref{algorithm12})} works in such a way that, initially, it checks whether $ p^{2} > t $, because for such parameters it can execute the \textsc{Two-Lists} 
algorithm which complexity if linear in $ t $ (see Fact
\ref{fact12}). If this is not the case, the main body of the algorithm is executed, yet another specific condition is checked: $ \log_{2}(p) > e^{\frac{\sqrt{t}}{32}} $. If so, it assigns all tasks
to each station. If every station has all the tasks assigned, then after $ t $ phases we may be sure that all the tasks are done, because always at least one station remains operational.
Because of the specific range of the parameters, the redundant work in this case is acceptable (i.e., within the claimed bound) from the point of view of our analysis. 
However we execute procedure \textsc{Confirm-Work} \mj{(Algorithm \ref{algorithm13})} in order to confirm this fact on the channel.

\textsc{Confirm-Work} is a type of leader election procedure. It assigns certain probability of a station to broadcast, in a way expecting that exactly one station will transmit in a number of trials.
Because we cannot be sure what is the actual number of operational stations, the probability is changed multiple times until all the tasks are confirmed.
}

If the specific conditions are discussed above are not satisfied (Algorithm \ref{algorithm12} lines 1-10), 
then \textsc{Mix-And-Text} is executed \mj{(Algorithm \ref{algorithm11})}. It changes the order of stations on list $ \texttt{STATIONS} $. Precisely, stations that performed successful broadcasts are moved to front of that list.
This procedure has two purposes. On one hand changing the order makes the adversary less flexible in crashing stations, as its order is already determined. On the other hand, we may predict with high probability
to which interval $n \in(\frac{p}{2^{i}}, \frac{p}{2^{i-1}}] $ for $ i = 1, \cdots, \lceil\log_{2}(p)\rceil $ does the current number of operational stations belong, what is important from the work analysis
perspective.

\begin{figure}[htb]
 \begin{center}
 \includegraphics[scale=0.6]{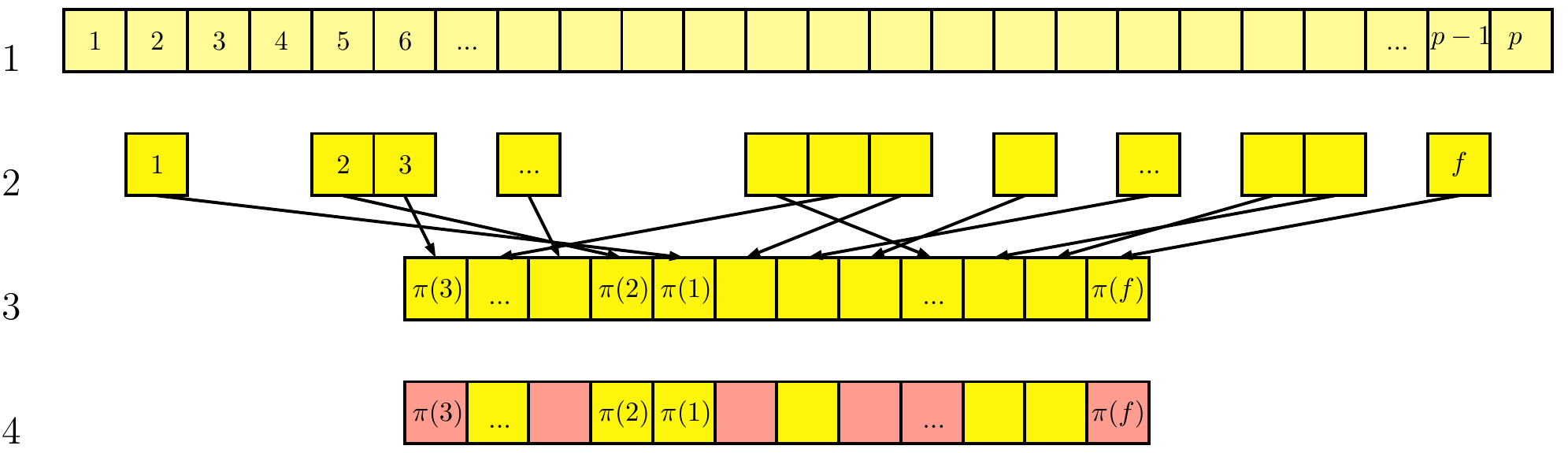}
 \caption{(1) Initially we have $ p $ stations. (2) The adversary chooses $ f $ stations prone to crashes. (3) Then it declares the order according to which the stations will crash. (4) \textsc{Mix-And-Test}
 chooses a number of leaders which are expected to be distributed uniformly among the adversary linear order.}
 \label{fig:doall7}
 \end{center}
\end{figure}

Stations moved to front of list \texttt{STATIONS} are called \textit{leaders}. \textit{Leaders} are chosen in a random process, so we expect that they will be uniformly distributed in the
adversary order between stations that were not chosen as \textit{leaders}. This allows us to assume that a crash of a \textit{leader} is likely to be preceded by several crashes of other stations
(see Figure \ref{fig:doall7}).

Let us consider procedure \textsc{Mix-And-Test} in detail. If $ n $ is the previously predicted number of operational stations, then each of the stations tosses a coin with the probability of success equal
$ 1/n $. In case where none or more than one of the stations broadcasts then silence is heard on the channel, as there is no collision detection. Otherwise, when only one station did successfully broadcast
it is moved to front of list $ \texttt{STATIONS} $ and the procedure starts again with a decremented parameter. However stations that have already been moved to front do not take part
in the following iterations of the procedure.

\jm{
Upon having chosen the leaders, regular work is performed. However, an important feature of our algorithm is that we do not perform full epochs, but only $ \sqrt{t} $ phases of each \textsc{Two-Lists} epoch. 
This allows us to be sure that the total work accrued in each epoch does not exceed $ p\sqrt{t} $. If, at some point, the number of successful broadcasts substantially drops, another \textsc{Mix-And-Test} 
(Algorithm \ref{algorithm11}) procedure is executed and a new set of leaders is chosen.
}

Before the algorithm execution the Linearly-Ordered adversary has to choose $ f $ stations prone to crashes and declare an order that will describe in what order those crashes may happen.
In what follows, when there are unsuccessful broadcasts of \textit{leaders} (crashes) we may be approaching the case when $ n \leq \sqrt{t} $ and we can execute \textsc{Two-Lists} that complexity is
linear in $ t $ for such parameters. Alternatively the adversary spends the majority of its possible crashes and the stations may finish all the tasks without any distractions.

\subsection{Analysis of ROBAL}

\jm{We begin our analysis with a general statement about the reliability of \Robal.}

\begin{lemma}
 Algorithm ROBAL is reliable.
 \label{fact-to-lem11}
\end{lemma}

\begin{proof}

We need to show that all the tasks will be performed as a result of executing the algorithm. First of all, if we fall in to the case when $ \frac{p}{2^{i}} \leq \sqrt{t} $ (or initially $ p \leq \sqrt{t} $) then
\textsc{Two-Lists} is executed, which is reliable as we know from \cite{CKL}.

Secondly, when $\log_{2}(p) > e^{\frac{\sqrt{t}}{32}} $ we assign all the tasks to every station and let the stations work for $ t $ phases. We know that $ f < p $ so at least one station will
perform all the tasks.

Finally, if those conditions do not hold, the algorithm runs an external loop in which variable $ i $ increments after each iteration. If the loop is performed $ \lceil \log_{2}(p) \rceil $ times 
then we run \textsc{Two-Lists}. Variable $ i $ may not be incremented only if the algorithm will enter and stay in the internal loop. However this is possible only after performing all the tasks, because the
internal loop runs for a constant number of times until all tasks are completed.
\end{proof}

\jm{\noindent We now proceed to a statement bounding the worst case work of \textsc{Two-Lists}, which is used as a sub-procedure in \Robal.}

\begin{fact}
\textsc{Two-Lists} always solves the Do-All problem with $ \mathcal{O}(pt) $ work.
 \label{fact12}
\end{fact}
$\mathcal{O}(pt) $ work is consistent with a scenario when every station performs every task. Comparing it with how \textsc{Two-Lists} works, justifies the fact.

\jm{We already mentioned that \textsc{ROBAL} was modeled in such a way that, whenever $ \frac{p}{2^{i}}\leq \sqrt{t} $ holds, the \textsc{Two-Lists} algorithm is executed, because
the work complexity of \textsc{Two-Lists} for such parameters is $ \mathcal{O}(t) $. We will prove it in the following fact.}

\begin{fact}
 Let $ n $ be the number of operational processors, and $ t $ be the number of outstanding tasks. Then for $ n \leq \sqrt{t} $ \textsc{Two-Lists} work complexity is $ \mathcal{O}(t) $.
\label{fact13}
 \end{fact}

\begin{proof}
If $ n \leq \sqrt{t} $, then the outstanding number of crashes is $ f < n $, hence $ f < \sqrt{t} $.
Algorithm \textsc{Two-Lists} has $ \mathcal{O}(t + p\sqrt{t} + p\:\min\{f, t\}) $ work complexity. 
In what follows the complexity is $ \mathcal{O}(t + \sqrt{t}\sqrt{t} + \sqrt{t}\:\min\{\sqrt{t}, t\}) $
$ = \mathcal{O}(t) $.
\end{proof}

\remove{
\begin{lemma}
 (\cite{PC}, Theorem 4.5) Let $ X_{1} + \cdots + X_{n} $ be independent Poisson trials, such that $ \mathbb{P}[X_{i}] = p_{i} $. Let $ X = \sum_{i = 1}^{n} X_{i}$, and $ \mu = \mathbb{E}X $.
 Then the following Chernoff bound holds:
 
  $$ \mathbb{P}[X \leq (1-\epsilon)\mu] \leq e^{-\frac{\epsilon^{2}}{2} \mu }, $$ for $ 0 < \epsilon < 1 .$
\label{lem11}
\end{lemma}
}

\jm{\noindent Figure \ref{fig:doall7} presents the way how we expect leaders to interlace other stations in the adversary's order. The following lemma estimates the probability that if a number of leaders
was crashed, then, overall, a significant number of stations must have been crashed as well.}

\begin{lemma}

Let us assume that we have $ n $ operational stations at the beginning of an epoch, where $ \sqrt{t} $ were chosen leaders.
 If the adversary crashes $ n/2 $ stations, then the probability that there were $ 3/4 $ of the overall number of leaders crashed in this group does not exceed
 $ e^{-\frac{1}{8} \sqrt{t}} $.
 \label{lem12}
\end{lemma}

\begin{proof}
 We have $ n $ stations, among which $ \sqrt{t} $ are leaders.
 The adversary crashes $ n/2 $ stations and our question is how many leaders where in this group?
 
 The hypergeometric distribution function with parameters $ N $ - number of elements, $ K $ - number of highlighted elements, $ l $ - number of trials, $ k $ - number of successes, is given by:
 
 $$ \mathbb{P}[X = k] = \frac{\binom{K}{k}\binom{N-K}{l-k}}{\binom{N}{l}}. $$
 
 \noindent
 The following tail bound from \cite{Hoef} tells us, that for any $ t > 0 $ and $ p = \frac{K}{N} $:
 
 $$ \mathbb{P}[X \geq (p + t)l] \leq e^{-2t^{2}l}. $$
 
 \noindent
 Identifying this with our process we have that $ K = n/2 $, $ N = n $, $ l = \sqrt{t} $ and consequently $ p = 1/2 $. Placing $ t = 1/4 $ we have that
 
 $$ \mathbb{P}\left[X \geq \frac{3}{4}\sqrt{t}\right] \leq e^{-\frac{1}{8}\sqrt{t}}. $$
\end{proof}

\jm{\noindent The following two lemmas give us the probability that \textsc{Mix-And-Test} diagnoses the number of operational stations properly, and hence, that the whole randomized part of \Robal\ works correctly
with high probability.}

\begin{lemma}
 Let us assume that the number of operational stations is in $(\frac{p}{2^{i}}, \frac{p}{2^{i-1}}] $ interval. Then procedure
 \textsc{Mix-And-Test}($i, t, p$) will return \textit{true} with probability $ 1 - e^{-c\;\sqrt{t}\log_{2}(p)} $, for some $ 0 < c < 1 $.
 \label{lem14}
\end{lemma}

\begin{proof}

\begin{claim} \label{claim11}
 Let the current number of operational stations be in $ (\frac{x}{2}, x] $.
 Then the probability of an event that in a single iteration of \textsc{Mix-And-Test} exactly one station will broadcast 
 is at least $ \frac{1}{2\sqrt{e}} $ (where the $\texttt{coin}^{-1}$ parameter is $ \frac{1}{x} $).
\end{claim}

\begin{proof}
Let us consider a scenario where the number of operational stations is in $ (\frac{x}{2}, x] $ for some $ x $. 
If every station broadcasts with probability of success equal $ 1/x $ then the probability of an event that exactly one station will transmit is $ (1 - \frac{1}{x})^{x-1} \geq 1/e $. 
Estimating the worst case, when there are $ \frac{x}{2} $ living stations (and the probability of success remains $ 1/x $) we have that
$$ \frac{1}{2} \left(1 - \frac{1}{x}\right)^{x \cdot \frac{x-2}{2}} \geq \frac{1}{2\sqrt{e}}. $$
\end{proof}

According to Claim \ref{claim11} the probability of an event that in a single round of \textsc{Mix-And-Test} exactly one stations will be heard is $ \frac{1}{2\sqrt{e}} $.

We assume that $ n \in (\frac{p}{2^{i}}, \frac{p}{2^{i-1}}] $. We will show that the algorithm confirms appropriate $ i $ with probability $ 1 - e^{-c\;\sqrt{t}\log_{2}p} $. 
For this purpose we need $ \sqrt{t} $ transmissions to be heard. 

Let $ X $ be a random variable such that $ X = X_{1} + \cdots + X_{\sqrt{t}\log_{2}(p)}, $ where $ X_{1}, \cdots, X_{\sqrt{t}\log_{2}(p)} $ are Poisson trials and

$$
X_{k} = \left\{ \begin{array}{ll}
1 & \textrm{if station broadcasted,}\\
0 & \textrm{otherwise.}
\end{array} \right 
.$$
We know that $$ \mu = \mathbb{E}X = \mathbb{E}X_{1} + \cdots + \mathbb{E}X_{\sqrt{t}\log_{2}(p)} \geq \frac{\sqrt{t}\log_{2}(p)}{2\sqrt{e}}. $$
To estimate the probability that $ \sqrt{t} $ transmissions were heard we will use the Chernoff's inequality.

We want to have that $ (1-\epsilon)\mu = \sqrt{t} $. Thus $ \epsilon = \frac{\mu - \sqrt{t}}{\mu} = \frac{\log_{2}(p) - 2\sqrt{e}}{\log_{2}(p)} $ and
$ 0 < \epsilon < 1 $ for sufficiently large $ p $.
Hence

$$ \mathbb{P}[X < \sqrt{t}] \leq e^{-\frac{\left(\frac{\log_{2}(p) - 2\sqrt{e}}{\log_{2}(p)}\right)^{2}}{2}  \frac{\sqrt{t}\log_{2}(p)}{2\sqrt{e}}} = e^{-c\;\sqrt{t}\log_{2}(p)}, $$
for some bounded $ 0 < c < 1 $. We conclude that with probability $ 1 - e^{-c\;\sqrt{t}\log_{2}(p)} $ we will confirm the correct $ i $ which describes and estimates 
the current number of operational stations.
\end{proof}

\begin{lemma}
 $\textsc{Mix-And-Test}(i, t, p) $ will \textbf{not be} executed if there are more than $ \frac{p}{2^{i-1}} $ operational stations, with probability not less than  
 $ 1 - (\log_{2}(p))^{2}\:\max\{ e^{-\frac{1}{8}\sqrt{t}} ,e^{-c\;\sqrt{t}\log_{2}(p)} \}. $
 \label{lem16}
\end{lemma}

\begin{proof}
Let $ A_{i} $ denote an event that at the beginning of and execution of the \textsc{Mix-And-Test}($i, t, p$) procedure there are no more than $ \frac{p}{2^{i-1}} $ operational stations.

The basic case then $ i = 0 $ is trivial, because initially we have $ p $ operational stations, thus $ \mathbb{P}(A_{0}) = 1 $. Let us consider an arbitrary $ i $.
We know that 

$$ \mathbb{P}(A_{i}) = \mathbb{P}(A_{i}|A_{i-1})\mathbb{P}(A_{i-1}) + \mathbb{P}(A_{i}|A^{c}_{i-1})\mathbb{P}(A^{c}_{i-1}) \geq \mathbb{P}(A_{i}|A_{i-1})\mathbb{P}(A_{i-1}). $$
Let us estimate $ \mathbb{P}(A_{i}|A_{i-1}) $. Conditioned on that event $ A_{i-1} $ holds, we know that after executing \textsc{Mix-And-Test}($i-1, t, p$) we had $ \frac{p}{2^{i-2}} $ operational stations.
In what follows if we are now considering \textsc{Mix-And-Test}($i, t, p$), then we have two options:

\begin{enumerate}
 \item \textsc{Mix-And-Test}($i-1, t, p$) returned \textit{false},
 \item \textsc{Mix-And-Test}($i-1, t, p$) returned \textit{true}.
\end{enumerate}
Let us examine what do these cases mean:
\begin{enumerate}
 \item[Ad 1.] If the procedure returned \textit{false} then we know from Lemma \ref{lem14} that with probability $ 1 - e^{-c\;\sqrt{t}\log_{2}(p)} $ there had to be no more than $ \frac{p}{2^{i-1}} $ operational
 stations. If that number would be in  $ (\frac{p}{2^{i-1}}, \frac{p}{2^{i-2}}] $ then the probability of returning \textit{false} would be less than $ e^{-c\;\sqrt{t}\log_{2}(p)} $.
 \item[Ad 2.] If the procedure returned \textit{true}, this means that when executing it with parameters $ (i-1, f, p) $ we had no more than $ \frac{p}{2^{i-1}} $ operational stations.
 Then the internal loop of \textsc{ROBAL} was broken, so according to Lemma \ref{lem12} we conclude that the overall number of operational stations had to reduce by half with probability at least
 $ 1 - e^{-\frac{1}{8}\sqrt{t}} $.
\end{enumerate}
Consequently, we deduce that $ \mathbb{P}(A_{i}|A_{i-1}) \geq (1 - \max\{ e^{-\frac{1}{8}\sqrt{t}} ,e^{-c\;\sqrt{t}\log_{2}(p)} \}) $. Hence
$ \mathbb{P}(A_{i}) \geq (1 - \max\{ e^{-\frac{1}{8}\sqrt{t}} ,e^{-c\;\sqrt{t}\log_{2}(p)} \})^{i} $. Together with the fact, that $ i \leq \log_{2}(p) $ and the Bernoulli inequality we have that

$$ \mathbb{P}(A_{i}) \geq 1 - \log_{2}(p)\:\max\{ e^{-\frac{1}{8}\sqrt{t}} ,e^{-c\;\sqrt{t}\log_{2}(p)} \}.$$
We conclude that the probability that the conjunction of events $ A_{1},\cdots,A_{\log_{2}(p)} $ will hold is at least

$$ \mathbb{P}\left(\bigcap_{i = 1}^{\log_{2}(p)}A_{i}\right) \geq 1 - (\log_{2}(p))^{2}\:\max\{ e^{-\frac{1}{8}\sqrt{t}} ,e^{-c\;\sqrt{t}\log_{2}(p)} \}.$$

\end{proof}

\jm{\noindent We can now proceed to the main result of this section.}

\begin{theorem}
 ROBAL performs $\mathcal{O}(t + p\sqrt{t}\log(p)) $ expected work against the Linearly-Ordered adversary in the channel without collision detection.
 \label{theorem11}
\end{theorem}

\begin{proof}
In the algorithm we are constantly controlling whether condition $ \frac{p}{2^{i}} > \sqrt{t} $ holds. If not, then we execute \textsc{Two-Lists} which
complexity is $ \mathcal{O}(t) $ for such parameters.

If this condition does not hold initially then we check another one i.e. whether $ \log_{2}(p) > e^{\frac{\sqrt{t}}{32}} $ holds. For such configuration we assign all the tasks to every station.
The work accrued during such a procedure is $ \mathcal{O}(pt) $. However when $ \log_{2}(p) > e^{\frac{\sqrt{t}}{32}} $ then together with the fact that  $ e^{x} < x $ we have that $ \log_{2}(p) > t $ 
and consequently the total complexity is $ \mathcal{O}(p\log(p)) $.

Finally, the successful stations, that performed all the task have to confirm this fact. We demand that only one station will transmit and if this happens, the algorithm terminates.
The expected value of a geometric random variable lets us assume that this confirmation will happen in expected number of $ \mathcal{O}(\log(p)) $ rounds, generating $ \mathcal{O}(p\log(p)) $ work.

When none of the conditions mentioned above hold, we proceed to the main part of the algorithm.
The testing procedure by \textsc{Mix-And-Test} for each of disjoint cases, where $ n \in (\frac{p}{2^{i}}, \frac{p}{2^{i-1}}] $ requires a certain amount of work that can be estimated by
$ \mathcal{O}(p \sqrt{t}\log(p)) $, as there are $ \sqrt{t}\log_{2}(p) $ testing phases in each case and at most $ \frac{p}{2^{i}} $ stations take part in a single testing phase for a certain case.

In the algorithm we run through disjoint cases where $ n \in (\frac{p}{2^{i}}, \frac{p}{2^{i-1}}] $. 
From Lemma \ref{lem12} we know that when some of the leaders were crashed, then a proportional number of all the stations had to be crashed.
When leaders are crashed but the number of operational stations still remains in the same interval, then the lowest number of tasks will be confirmed if only the initial segment of stations will transmit.
As a result, when half of the leaders were crashed, then the system still confirms $ \frac{t}{8} = \Omega(t) $ tasks.
This means that even if so many crashes occurred, $ \mathcal{O}(1) $ epochs still suffice to do all the tasks. Summing work over all the cases may be estimated as
$ \mathcal{O}(p\sqrt{t}) $.

By Lemma \ref{lem16} we conclude that the expected work complexity is bounded by:
$$ \left((\log(p))^{2}\:\max\{ e^{-\frac{1}{8}\sqrt{t}} ,e^{-c\;\sqrt{t}\log(p)} \}\right)\mathcal{O}(pt + p\sqrt{t}\log^{2}(p)) $$
$$ + \left(1 - (\log(p))^{2}\:\max\{ e^{-\frac{1}{8}\sqrt{t}} ,e^{-c\;\sqrt{t}\log(p)} \}\right)\mathcal{O}(p\sqrt{t}\log(p)) = \mathcal{O}(p\sqrt{t}\log(p)), $$
where the first expression  comes from the fact, that if we entered the main loop of the algorithm then we know that we are in a configuration where $ \log_{2}(p) \leq e^{\frac{\sqrt{t}}{32}} $. 
Thus we have that
 
$$ \frac{pt + p\sqrt{t}\log^{2}(p)}{e^{\frac{\sqrt{t}}{8}}} \leq \frac{pt + pt\log^{2}(p)}{e^{\frac{\sqrt{t}}{16}}e^{\frac{\sqrt{t}}{16}}} \leq \frac{p + p\log^{2}(p)}{e^{\frac{\sqrt{t}}{16}}} 
\leq p + p\log(p) = \mathcal{O}(p\log(p)),  $$
which ends the proof.
\end{proof}

%% file: grubtech.tex
\section{\Grubtech~ --- Groups Together with Echo}
\label{grubtech}

In this section we present a randomized algorithm designed to reliably perform 
\DA\ 
in the presence of a Weakly-Adaptive adversary on a shared
channel without collision detection. Its expected work complexity is $ \mathcal{O}(t + p\sqrt{t} + p\;\min\{p/(p-f), t\}\log(p)) $.

\subsection{Description of \Grubtech}
Our solution is built on algorithm \textsc{Groups-Together} \mj{(details in Section \ref{grotog})} and a newly designed \textsc{Crash-Echo} procedure that works as a kind of fault-tolerant 
replacement of collision detection mechanism (which is not present in the model).
In fact, the algorithm presented here is asymptotically only logarithmically far from matching the lower bound shown in \cite{CKL}, which, to some extent, answers the open question stated therein.

\begin{algorithm}
{
{- initialize \texttt{STATIONS} to a sorted list of all $ p $ stations\;}
{- arrange all $ p $ names of stations into list $\texttt{GROUPS}$ of groups\;}
{- initialize both $\texttt{TASKS}$ and $\texttt{OUTSTANDING}_{v} $ to sorted list of all $ t $ names of tasks\;}
{- initialize $\texttt{DONE}_{v} $ to an empty list of tasks\;}
{- initialize $ i := 0 $ \;}
{- initialize $ leader $ := $ \textsc{Elect-Leader}(i) $ and add the $ leader $ to each group\;}
{- \Repeat{halted}{$\textsc{Epoch-Groups-CE}(i)$\;}}
}
\caption{\textsc{GrubTEch}; code for station $ v $}
\label{algorithm21}
\end{algorithm}

\begin{algorithm}
{
{set pointer $ \texttt{Task\_To\_Do}_{v} $ on list $ \texttt{TASKS} $ to the initial position of the range $ v $\;}
{set pointer $ \texttt{Transmit} $ to the first item on list $ \texttt{GROUPS} $\;}
{\Repeat{pointer $ \texttt{Transmit} $ points to the first entry on list $\texttt{GROUPS} $}
{

{\tcp{Round 1:}}
{perform the first task on list $ \texttt{TASKS} $, starting from the one pointed to by $ \texttt{Task\_To\_Do}_{v} $, that is in list $\texttt{OUTSTANDING}_{v} $\;}
{move the performed task from list $\texttt{OUTSTANDING}_{v} $ to list $\texttt{DONE}_{v} $\;}
{advance pointer $ \texttt{Task\_To\_Do}_{v} $ by one position on list $ \texttt{TASKS} $\;}
{\tcp{Rounds 2 \& 3:}}
{\If{$ \texttt{Transmit} $ points to $ v $}{execute $ \textsc{Crash-Echo} $\;} }
{attempt to receive a pair of messages\;}
{\tcp{Round 4:}}
{\If{(silent, loud) was heard in the preceding round}{
{let $ w $ be the first station in the group pointed to by $\texttt{Transmit}$:\;}
{\For{each item $ x $ on list $\texttt{DONE}_{w} $}{
{\If{$ x $ is on list $ \texttt{OUTSTANDING}_{v} $}{move $ x $ from $\texttt{OUTSTANDING}_{v} $ to $\texttt{DONE}_{v} $\;}}
{\If{$ x $ is on list $ \texttt{TASKS} $}{remove $ x $ from $ \texttt{TASKS} $\;}}
}}
{\If{list $ \texttt{TASKS} $ is empty}{halt\;}}
{advance pointer $ \texttt{Transmit} $ by one position on list $ \texttt{GROUPS} $\;}
}}\Else{

{\If{(loud, loud) was heard in the preceding round} {remove the group pointed to by $ \texttt{Transmit} $ from $ \texttt{GROUPS} $\;}
\Else{
{remove $ leader $ from all the groups on list $ \texttt{GROUPS} $\;
$ leader $ := $ \textsc{Elect-Leader}(i)$ and add the leader to each group\;}
}}
}
}}
{rearrange all stations in the groups of list $ \texttt{GROUPS} $ into a new version of list $ \texttt{GROUPS} $\;}}
\caption{Procedure \textsc{Epoch-Groups-CE}; code for station $ v $}
\label{algorithm22}
\end{algorithm}

\begin{algorithm}
{
{\tcp{Round 1:}}
{broadcast one bit\;}
{\tcp{Round 2:}}
{\If{$ v = leader $}{broadcast one bit\;}}
}
\caption{Procedure Crash-Echo; code for station $ v $}
\label{algorithm23}
\end{algorithm}

\begin{algorithm}
{
\KwIn{i}

\Repeat{$ i < p $}{
{$ \texttt{coin} := p $\;}
{toss a coin with the probability $ \texttt{coin}^{-1} $ of heads to come up\;}
\If{heads came up in the previous step}{broadcast $ v $ via the channel and attempt to receive a message\;}
\If{some station $ w $ was heard}{
$ leader := w $\;
return $ leader $\;
}
\Else{
increment $ i $ by $ 1 $\;}
}
{set pointer $ \texttt{Transmit}_{STATIONS} $ to the first item on list \texttt{STATIONS}\;}
\Repeat{a transmission was heard}{
\If{$ \texttt{Transmit}_{STATIONS} $ points to $ v $}{
broadcast one bit\;
attempt to receive a message\;
\If{some station $ w $ was heard}{
$ leader := w $\;
return $ leader $\;
}}
{advance pointer $ \texttt{Transmit}_{STATIONS} $ by one position on list $ \texttt{STATIONS} $\;}
}
}
\caption{Procedure Elect-Leader; code for station $ v $}
\label{algorithm24}
\end{algorithm}

\paragraph{The Crash-Echo procedure.}
\mj{Let us recall the details of \textsc{Groups-Together} from Section \ref{grotog}. All the stations within a certain group have the same tasks assigned
and when it comes to transmitting they do it simultaneously. This strongly relies on the collision detection mechanism, as the stations do not necessarily need to know which station transmitted, but they need to know
that there is progress in tasks performance. That is why if a collision is heard and all the stations within the same group were doing the same tasks, we can deduce that those tasks were actually done.}

In our model we do not have collision detection, however we designed a mechanism that provides the same feedback without contributing too much work to the algorithm's complexity.
Strictly speaking we begin with choosing a leader. His work will be of a dual significance. On one hand 
he will belong to some group and perform tasks regularly. But on the other hand he will also perform additional transmissions in order to indicate
whether there was progress when stations transmitted.

When a group of stations is indicated to broadcast the \textsc{Crash-Echo} procedure is executed.
It consists of two rounds where the whole group transmits together with the leader in the first one and in the second only the leader transmits.
We may hear two types of signals:

\begin{itemize}
 \item \textbf{loud} - a legible, single transmission was heard. Exactly one station transmitted.
 \item \textbf{silent} - a signal indistinguishable from the background noise is heard. None or more than one station transmitted.
\end{itemize}

\noindent Let us examine what are the possible pairs (group \& leader, leader) of signals heard in such approach:

\begin{itemize}
 \item \textbf{(silent, loud)} - in the latter round the leader is operational, so he must have been operational in the former round. Because silence was heard in the former 
 round this means that there was a successful transmission of at least two stations one of which was the leader. This is a fully successful case.
 \item \textbf{(loud, loud)} - the former and the latter round were loud, so we conclude that it was the leader who transmitted in both rounds. 
 If the leader belonged to the group scheduled to transmit, then we have progress; otherwise not.
 \item \textbf{(silent, silent)} - if both rounds were silent we cannot be sure was there any progress. Additionally we need to elect a new leader.
 \item \textbf{(loud, silent)} - when the former round was loud we cannot be sure whether the tasks were performed; a new leader needs to be chosen.
\end{itemize}

Nevertheless, the Weakly-Adaptive adversary has to declare some $ f $ stations that are prone to crashes. The elected leader might belong to that subset and be crashed
at some time. When this is examined, the algorithm has to switch to the \textsc{Elect-Leader} mode, in order to select another leader.
Consequently the most significant question from the point of view of the algorithm's analysis is what is the expected number of trials to choose a non-faulty leader.

\paragraph{Two modes.}

We need to select a leader and be sure that he is operational in order to have our progress indicator working instead of the collision detection mechanism. When the leader is operational we simply
run \textsc{Groups-Together} algorithm with the difference that instead of a simultaneous transmission by all the stations within a group, we run the \textsc{Crash-Echo} procedure that allows us to distinguish
whether there was progress.

Choosing the leader is performed by procedure \textsc{Elect-Leader}, where each station tosses a coin with the probability of success equal $ 1/p $. If a station is successful then it transmits in the following round.
If exactly one station transmits then the leader is chosen. Otherwise the experiment is continued (for $ p $ rounds in total). Nevertheless if this still does not work, then the first station that transmits
in a round-robin fashion procedure, becomes the leader.

\jm{Note that we have a special variable $ i $ used as a counter that is incremented in the \textsc{Elect-Leader} procedure until it reaches value $ p $. We assume that this value is passed to \textsc{Elect-Leader}
by reference, so that its incrementation is also recognized in the main body of the \textsc{Epoch-Groups-CE} algorithm, thus $ i $ is a global counter.}

%
%
%
%
%

\subsection{Analysis of GrubTEch}

Let us begin the analysis of \textsc{GrubTEch} by recalling an important result from \cite{CKL}.

\begin{theorem}
 (\cite{CKL}, Theorem 6) The Weakly-Adaptive $ f $-Bounded adversary can force any reliable randomized algorithm solving Do-All in the channel without collision detection to perform the expected work
 $ \Omega(t + p\sqrt{t} + p\;\min\{p/(p-f), t\}) $.
 \label{theorem22}
\end{theorem}
\noindent
In fact the theorem above in \cite{CKL} stated that the lower bound was $ \Omega(t + p\sqrt{t} + p\;\min\{f/(p-f), t\}) $, however the proof relied on the expected round in which the first successful 
transmission took place and the authors did not take into consideration that the first successful transmission may occur earliest in round $ 1 $.
Hence as it must be at least round number $ 1 $ we correct it as follows: $ \frac{f}{p-f} + 1 = \frac{f}{p-f} + \frac{p-f}{p-f} = \frac{p}{p-f}.$

\begin{lemma}
 \textsc{GrubTEch} is reliable.
\end{lemma}

\begin{proof}
 The reliability of \textsc{GrubTEch} is a consequence of the reliability of \texttt{Groups-Together}. We do not make any changes in the core of the algorithm.
 \texttt{Crash-Echo} does not affect the algorithm, as it always finishes. \texttt{Elect-Leader} procedure always finishes as well. The first loop is executed for at most $ p $ times and then it ends.
 The second loop awaits to hear a broadcast in a round-robin manner. But we know that $ 0 \leq f \leq p - 1 $, so always one processor remains operational and it will respond.
\end{proof}

\dk{Let us define a {\em sustainable leader} as a station that is operational until the end of the execution or a non-faulty station, and was elected as a leader during
some execution of procedure \textsc{Elect-Leader}.} 

\begin{lemma}
 \jm{The total number of rounds during which procedure \textsc{Elect-Leader} is run (possibly splitted into several executions) until electing a sustainable leader is $ \log(p) \frac{4p}{p-f} $ with probability at least $1 - \frac{1}{p} $.}
 \label{lemma22}
\end{lemma}

\begin{proof}
 \dk{Recall that procedure \textsc{Elect-Leader} could be called several times,
 until selecting a sustainable leader at the latest.
 The expected number of rounds needed to elect a sustainable leader during these calls is upper bounded by the time needed to hit the first non-faulty station by the executions of procedure \textsc{Elect-Leader}. Hence, in the remainder of the proof we estimate the total number of such rounds with probability at least $1-1/p$.}

 We have $ p $ stations from which $ f $ are prone to crashes. Hence we have $ p - f $ non-faulty stations. That is why the probability that a non-faulty one will respond
 in the election procedure is at least $ (p-f)/p $. We may observe that this probability will increase if we failed in previous executions. In fact, after $ f $ executions
 we may be sure to choose a non-faulty leader. However we will estimate the probability of our process by an event of awaiting the first success in a number of trials, as
 our process is stochastically dominated by such a geometric distribution process.
 
 We have a channel without collision detection, so exactly one station has to transmit in order to elect a leader. Let $ x $ be the actual number of operational stations. 
 The probability $ s $ of the event that a non-faulty station will be elected in the procedure
 may be estimated as follows:
 
 $$ \frac{p-f}{p} \left(1 - \frac{1}{p}\right)^{x-1} \geq \frac{p-f}{p} \left(1 - \frac{1}{p}\right)^{p-1} \geq \frac{p-f}{p}\cdot\frac{1}{4}\cdot\frac{p}{p-1} $$
 $$ = \frac{p-f}{4(p-1)} \geq \frac{p-f}{4p}. $$
 
 Let us estimate the probability of awaiting the first success in a number of trials. Let $ X \sim Geom((p-f)/4p) $.
 We know that for a geometric random variable with the probability of success equal $ s $:
 
 $$ \mathbb{P}(X \geq i) = (1 - s)^{i-1}. $$
 Applying this to our case with $ i = \frac{4p}{p-f}\log(p) + 1 $ we have that
 
 $$ \mathbb{P}\left(X \geq \frac{4p}{p-f}\log(p) + 1\right) = \left(1 - \frac{1}{\frac{4p}{p-f}} \right)^{\frac{4p}{p-f}\log(p)} \leq e^{-\log(p)} = \frac{1}{p}. $$
 Thus the probability of a complementary event is
 
 $$ \mathbb{P}\left(X < \frac{4p}{p-f}\log(p) + 1\right) > 1 - \frac{1}{p}. $$
\end{proof}

\begin{theorem} \label{theorem23}
 \textsc{GrubTEch} solves Do-All in the channel without collision detection with the expected work $ \mathcal{O}(t + p\sqrt{t} + p\;\min\{p/(p-f), t\}\log(p)) $ against the Weakly-Adaptive $ f $-Bounded adversary. 
\end{theorem}

\begin{proof}
 We may divide the work of \textsc{GrubTEch} to three components: productive, failing and the one reasoning from electing the leader.
 
 Firstly, the core of our algorithm is the same as \textsc{Groups-Together} with the difference that we have the \textsc{Crash-Echo} procedure that takes twice as many
 transmission rounds. According to Fact \ref{theorem21}, it is sufficient to estimate this kind of work as $ \mathcal{O}(t + p\sqrt{t})$.
 
 Secondly, there is some work that results from electing the leader. According to Lemma \ref{lemma22}, \dk{a sustainable}
 leader will be chosen within $ \frac{4p}{p-f}\log(p) $ \dk{rounds of executing}
 \textsc{Elect-Leader} with high probability. That is why the expected work to elect a non-faulty leader is overall $ \mathcal{O}(p\;\frac{p}{p-f}\log(p)) $
 
 Finally, there is some amount of failing work that results from rounds where the \textsc{Crash-Echo} procedure indicated that the leader was crashed. However work accrued
 during such rounds will not exceed the amount of work resulting from electing the leader, hence we state that failing work contributes $ \mathcal{O}(p\;\frac{p}{p-f}\log(p))$ as well.
 
 Consequently, we may estimate the expected work of \textsc{GrubTEch} as
 
 $$ \left(1 - \frac{1}{p}\right)\mathcal{O}(t + p\sqrt{t} + p\;\min\{p/(p-f), t\}\log(p)) + \frac{1}{p}\mathcal{O}(p^{2}) $$
 $$ = \mathcal{O}(t + p\sqrt{t} + p\;\min\{p/(p-f), t\}\log(p))$$
 what ends the proof.
\end{proof}

%% file: grubtechpartord.tex
\section{How GrubTEch works for other partial orders}
\label{grubtechpartord}

The line of investigation originated by \textsc{ROBAL} and \textsc{GrubTEch} leads to a natural question whether considering some intermediate partial orders of the adversary
may provide different work complexities. In this section we answer this question in the positive by examining the \textsc{GrubTEch} algorithm against the k-Chain-Ordered adversary 
on a channel without collision detection.

\subsection{The lower bound}

\begin{theorem}
\label{t:lb}
For any reliable randomized algorithm solving \DA\ on the shared channel
and any integer $0<k\le f$, 
there is a $k$-chain-based partial order of $f$ elements 
such that the ordered adversary restricted by this order
can force the algorithm to perform the
expected work $\Omega(t+p\sqrt{t}+p\min\{k,f/(p-f),t\})$.
\end{theorem}

\begin{proof}
The part $\Omega(t+p\sqrt{t})$ follows from the absolute lower bound
on reliable algorithms on shared channel.
We prove the remaining part of the formula.
If $k > c\cdot f/(p-f)$, for some constant $0<c<1$, then that part is 
asymptotically dominated by $p\min\{f/(p-f),t\}$
and it is enough to take the order being an anti-chain of $f$
elements; clearly it is a $k$-chain-based partial order of $f$ elements,
and the adversary restricted by this order is equivalent to the weakly-adaptive
adversary, for which the lower bound $\Omega(p\min\{f/(p-f),t\})$
follows directly from Theorem \ref{theorem22}.
Therefore, in the reminder of the proof, assume $k\le c\cdot f/(p-f)$.

Consider the following strategy of the adversary in the first $\tau$ rounds, for some
value $\tau$ to be specified later.
Each station which wants to broadcast alone in a round is 
crashed in the beginning of this round, just before its intended transmission.
Let $\cF$ be the family of all subsets of stations containing $k/2$ elements.
Let $\cM$ denote the family of all partial orders consisting of $k$ independent
chains of roughly (modulo rounding) $f/k$ elements each.
Consider the first $\tau=k/2$ rounds.
The probability $\Pr(F)$, for $F\in \cF$, is defined to be equal to
the probability of an occurrence of an execution during the experiment, in
which exactly the stations with from set $F$ 
are failed by round~$\tau$.
Consider an order $M$ selected uniformly at random from $\cM$.
The probability that all elements of set $F\in \cF$ are in $M$
is a non-zero constant.
It follows from the following three observations. 
First, under our assumption, $k<f$ 
(as $k\le c\cdot f/(p-f)$ for some $0<c<1$).
Second, from the proof of the lower bound in~\cite{CKL} with respect to sets of size $O(f)$,
the probability is a non-zero constant provided in each round we have at most $c'\cdot f$
crashed processes, for some constant $0<c'<1$.
Third, since each successful station can enforce the adversary to fail at most
one chain, after each of the first $\tau=k/2$ rounds there are still at least
$k/2$ chains without any crash, hence at most $f/2$ crashes have been enforced
and the argument from the lower bound in \cite{CKL} could be applied.
To conclude the proof, non-zero probability of not hitting any element not in $M$
means that there is such $M\in \cM$ that the algorithm does not finish before round 
$\tau$ with constant probability, thus imposing expected work $\Omega(pk)$.
\end{proof}

\subsection{GrubTEch against the $k$-Chain-Ordered adversary}

The analysis of \textsc{GrubTEch} against the Weakly-Adaptive adversary relied on electing a leader.
Precisely, as we knew that there are $ p - f $ non-faulty stations in an execution, then we expected to elect a non-faulty leader in a certain number of trials.

Nevertheless we could have chosen a faulty station as a leader and the adversary could have chosen to crash that station. However the amount of such failing occurrences would not exceed
the number of trials needed to elect the non-faulty one. While considering the $k$-Chain-Ordered adversary, these estimates are different.

When a leader is elected then he may belong to the non-faulty set (and this is expected to happen within a certain number of trials) or he may be elected from the faulty set, thus will
be placed somewhere in the adversary's partial order. If the leader was elected in a random process then it will appear in a random part of this order. In what follows we may expect
that if the adversary decides to crash the leader, then it will be forced to crash several stations preceding the leader in one of the chains in his partial order. 
Consequently this is the key reason why the expected work complexity would change against the k-Chain-Ordered adversary.

\begin{theorem}
\label{thmChainOrder}
 \textsc{GrubTEch} solves Do-All in the channel without collision detection with the expected work $ \mathcal{O}(t + p\sqrt{t} + p\;\min\{p/(p-f), k, t\}\log(p)) $ against the $k$-Chain-Ordered adversary.
\end{theorem}

\begin{proof}
Because of the same arguments as in Theorem \ref{theorem23}, it is expected that a non-faulty leader will be chosen in the expected number of $ \mathcal{O}(\frac{p}{p-f}\log(p)) $ trials, 
generating $ \mathcal{O}(p\;\frac{p}{p-f}\log(p)) $ work.

On the opposite, let us consider what will be the work accrued in phases when the leader is chosen from the faulty set and hence may be crashed by the adversary. 
According to the adversary's partial order we have initially $ k $ chains, where chain $ j $ has length $ l_{j} $. If the leader was chosen from that order then it belongs to one of the chains. 
We will show that it is expected that the chosen leader will be placed somewhere in the middle of that chain.

Let $ X $ be a random variable such that $ X_{j} = i $ where $ i $ represents the position of the leader in chain $ j $. We have that
$ \mathbb{E}X_{j} = \sum_{i=1}^{l_{j}} \frac{i}{l_{j}} = \frac{1}{l_{j}} \frac{(1+l_{j})}{2}l_{j} = \frac{(1+l_{j})}{2}. $

We can see that if the leader was crashed, this implies that half of the stations forming the chain were also crashed.
If at some other points of time, the faulty leaders will also be chosen from the same chain, then by simple induction we may conclude that this chain is expected to be all crashed after $ \mathcal{O}(\log(p)) $ 
iterations, as a single chain has length $ \mathcal{O}(p) $ at most. In what follows if there are $ k $ chains, then after $ \mathcal{O}(k\log(p)) $ steps this process will end and we may be sure to choose 
a leader from the non-faulty subset, because the adversary will spend all his failure possibilities.

Finally, if we have a well serving non-faulty leader then the work accrued is asymptotically the same as in \textsc{Groups-Together} algorithm with the difference that each step 
is now simulated by the \textsc{Crash-Echo} procedure. This work is equal $ \mathcal{O}(t + p\sqrt{t}) $.

Altogether, taking Lemma \ref{lemma22} into consideration, the expected work performance of \textsc{GrubTEch} against the k-Chain-Ordered adversary is  
$$ \left(1 - \frac{1}{p}\right)\mathcal{O}(t + p\sqrt{t} + p\;\min\{p/(p-f), k, t\}\log(p)) + \frac{1}{p}\mathcal{O}(p^{2}) $$
$$ = \mathcal{O}(t + p\sqrt{t} + p\;\min\{p/(p-f), k, t\}\log(p))$$
what ends the proof.

\end{proof}

\subsection{GrubTEch against the  adversary limited by arbitrary order}

Finally, let us consider the adversary that is limited by \textbf{arbitrary} partial order $P=(P,\succ)$. We say that  two partially ordered elements  
are \textit{incomparable} if none of relations $x \succ y $ and  $y\succ x$ hold. Translating into the considered model, this means that the adversary may 
crash incomparable elements in any  sequence during the execution of the algorithm (clearly, only if $x$ and $y$ are among  $f$  stations chosen to be crash-prone). 
%

\begin{theorem}
\textsc{GrubTEch} 
solves Do-All in the channel without collision detection with the expected work $ \mathcal{O}(t + p\sqrt{t} + p\;\min\{p/(p-f), k, t\}\log(p)) $ 
against the \dk{$k$-Thick-Ordered adversary}. 
\end{theorem}

\begin{proof}
We assume that the crashes forced by the adversary are constrained  by some partial order $P$. 
Let us first recall the following lemma. 

\begin{lemma}\label{DWlemma}(Dilworth's theorem \cite{DILWORTH})
In a finite partial order, the size of a maximum anti-chain is equal to the minimum
number of chains needed to cover all elements of the partial order.
\end{lemma}

\dk{Recall that the $k$-Thick-Ordered adversary is
constrained   by any order  of thickness~$k$.} 
Clearly, the adversary choosing some $f$ stations to be crashed cannot increase the size of the maximal anti-chain. 
Thus using Lemma~\ref{DWlemma} we consider the coverage  of the crash-prone stations by at most $k$ disjoint chains, \dk{and any dependencies between chains' elements
create additional constraints to the adversary comparing to the $k$-Chain-Ordered
one. Hence} we fall into the case 
concluded in Theorem~\ref{thmChainOrder} that completes the proof. 
\end{proof}

%% file: gilet.tex
\section{GILET --- Groups with Internal Leader Election Together} \label{gilet}

In this section we introduce an algorithm for the channel without collision detection that is designed to work efficiently against the $1$-RD adversary.
Its expected work complexity is $ \mathcal{O}(t + p\sqrt{t}\log^{2}(p)) $. The algorithm makes use of previously designed solutions from \cite{CKL}, i.e., \textsc{Groups-Together} algorithm \jmii{(cf. Section \ref{grotog})},
however we implement a major change in how the stations confirm their work
(due to the lack of collision detection in the model).

\begin{algorithm}
{
{- arrange all $ p $ names of stations into list $\texttt{GROUPS}$ of groups\;}
{- initialize variable k := $ p/\min\{\lceil\sqrt{t}\rceil, p\}$\;}
{- initialize both $\texttt{TASKS}$ and $\texttt{OUTSTANDING}_{v} $ to sorted list of all $ t $ names of tasks\;}
{- initialize $\texttt{DONE}_{v} $ to an empty list of tasks\;}
{- initialize $\texttt{REMOVED}$ to an empty list of stations\;}
{- \Repeat{halted}{\textsc{Epoch-Groups-CW($k$)}\;}}
}
\caption{\textsc{GILET}; code for station $ v $;}
\label{algorithm41}
\end{algorithm}

\begin{algorithm}
\KwIn{$k$}
{
{set pointer $ \texttt{Task\_To\_Do}_{v} $ on list $ \texttt{TASKS} $ to the initial position of the range $ v $\;}
{set pointer $ \texttt{Transmit} $ to the first item on list $ \texttt{GROUPS} $\;}
{\Repeat{pointer $ \texttt{Transmit} $ points to the first entry on list $\texttt{GROUPS} $}
{

{perform the first task on list $ \texttt{TASKS} $, starting from the one pointed to by $ \texttt{Task\_To\_Do}_{v} $, that is in list $\texttt{OUTSTANDING}_{v} $\;}
{move the performed task from list $\texttt{OUTSTANDING}_{v} $ to list $\texttt{DONE}_{v} $\;}
{advance pointer $ \texttt{Task\_To\_Do}_{v} $ by one position on list $ \texttt{TASKS} $\;}

{
\If{$ \texttt{Transmit} $ points to $ v $}{
{initialize $ i := 0 $ \;}
\Repeat{$ i < 4\log(p) $}{
execute $ \textsc{Mod-Confirm-Work} $($k$) \;
\If{a broadcast was heard}{break\;}
\Else{increment $ i $ by $ 1 $\;}}}
}

{\If{a broadcast was heard in the preceding round}{
{let $ w $ be the first station in the group pointed to by $\texttt{Transmit}$\;}
{\For{each item $ x $ on list $\texttt{DONE}_{w} $}{
{\If{$ x $ is on list $ \texttt{OUTSTANDING}_{v} $}{move $ x $ from $\texttt{OUTSTANDING}_{v} $ to $\texttt{DONE}_{v} $\;}}
{\If{$ x $ is on list $ \texttt{TASKS} $}{remove $ x $ from $ \texttt{TASKS} $\;}}
}}
{\If{list $ \texttt{TASKS} $ is empty}{halt\;}}
{advance pointer $ \texttt{Transmit} $ by one position on list $ \texttt{GROUPS} $\;}
}}
{\Else{
add all the stations from group pointed to by $\texttt{Transmit}$ to list $ \texttt{REMOVED} $\;
remove the group pointed to by $ \texttt{Transmit} $ from list $ \texttt{GROUPS} $\;
execute $\textsc{Check-Outstanding}$\;
halt\;}
}}
{rearrange all stations in the groups of list $ \texttt{GROUPS} $ into a new version of list $ \texttt{GROUPS} $\;}}}
\caption{Procedure \textsc{Epoch-Groups-CW}; code for station $ v $;}
\label{algorithm42}
\end{algorithm}

\begin{algorithm}
\KwIn{$k$}
{
{$ j := 0 $\;}
\Repeat{$ j < \log(k) $}{
{$ \texttt{coin} := \frac{k}{2^{j}} $\;}
{toss a coin with the probability $ \texttt{coin}^{-1} $ of heads to come up\;}
\If{heads came up in the previous step}{broadcast $ v $ via the channel and attempt to receive a message\;}
\If{some station $ w $ was heard}{
break\;
}
increment $ j $ by $ 1 $\;}
}
\caption{Procedure \textsc{Mod-Confirm-Work}; code for station $ v $}
\label{algorithm43}
\end{algorithm}

\begin{algorithm}
{
{- basing on list $\texttt{REMOVED}$ and list $\texttt{TASKS}$ assign every task to all the processors\;}
{- $i := 0$\;}
{\Repeat{$ i < |\texttt{TASKS}|$}{perform $ i $-th task from list \texttt{TASKS}\;
$i := i+1$\;}}
{- clear list $\texttt{TASKS}$\;}
}

\caption{Procedure \textsc{Check-Outstanding}; code for station $ v $}
\label{algorithm44}
\end{algorithm}

In our model, there is a channel without collision detection. That is why whenever some group $ g $ is scheduled to broadcast, a leader election procedure \textsc{Mod-Confirm-Work}
is executed in order to hear a successful transmission of exactly one station. Because all the stations within $ g $ had the same tasks assigned, then if the leader is chosen, 
we know that the group performed appropriate tasks.

The inherent cost of such an approach of confirming work is that we may not be sure whether removed groups did really crash.
The effect is that if all the tasks were not performed and all the stations were found crashed, then we have to execute an additional procedure that will finish performing them reliably.

This is realized by a new list \texttt{REMOVED} containing removed stations, and procedure ~\textsc{Check-Outstanding} which assigns every outstanding task to all the stations.
Then if with small probability we have mistakenly removed some operational stations, the algorithm still remains reliable and efficient.

\subsection{Analysis of GILET}

\begin{lemma}
\label{lemma41}
\textsc{GILET} is reliable.
\end{lemma}

\begin{proof}
 As well as in case of \textsc{GrubTEch}, the solution does depend on reliability of algorithm \textsc{Groups-Together}, because procedure
 \textsc{Mod-Confirm-Work} always terminates. If we fall into a mistake that some operational station has been removed from list \texttt{GROUPS}, then we
 execute procedure \textsc{Check-Outstanding} that will finish all the outstanding tasks.
\end{proof}

\begin{lemma}
\label{lem42}
Assume that the number of operational stations within a group is in $(\frac{k}{2^{i+1}}, \frac{k}{2^{i}}] $ interval and the $ \texttt{coin} $ parameter is set to $ \frac{k}{2^{i}} $.
 Then during $ \textsc{Mod-Confirm-Work} $ a confirming-work broadcast will be performed with probability at least $ 1 - \frac{1}{p} $.
\end{lemma}

\begin{proof}
 We assume that the number of operational stations is in $ (\frac{k}{2^{i+1}}, \frac{k}{2^{i}}] $. The probability that exactly one station will broadcast, estimated from the worst case point of
 view where only $ \frac{k}{2^{i+1}} $ stations are operational is $ \frac{1}{2\sqrt{e}} $, because of the same reason as in Claim \ref{claim11} of Lemma \ref{lem14}.
 
 That is why we would like to investigate the first success occurrence in a number of trials with the probability of success equal $ \frac{1}{2\sqrt{e}} $.
 
 Let $ X \sim Geom\left(\frac{1}{2\sqrt{e}}\right) $. We know that for a geometric random variable with the probability of success equal $ s $:
 
 $$ \mathbb{P}(X \geq i) = (1 - s)^{i-1}. $$
 Hence we will apply it for $ i = 2\sqrt{e}\log(p) + 1 $. We have that
 
 $$ \mathbb{P}(X \geq 2\sqrt{e}\log(p) + 1) = \left(1 - \frac{1}{2\sqrt{e}}\right)^{2\sqrt{e}\log(p)} \leq e^{-\log(p)} = \frac{1}{p}. $$
 Thus
 $$ \mathbb{P}(X > 2\sqrt{e}\log(p) + 1) > 1 - \frac{1}{p}. $$
\end{proof}

\begin{theorem}
\label{theorem41}
\textsc{GILET} performs $ \mathcal{O}(t + p\sqrt{t}\log^{2}(p)) $ expected work on channel without collision detection against the $1$-RD adversary.
\end{theorem}

\begin{proof}
 The proof of \textsc{Groups-Together} work performance from \cite{CKL} stated that noisy sparse epochs contribute $ \mathcal{O}(t) $ to work and silent sparse epochs contribute $ \mathcal{O}(p\sqrt{t}) $.
 Dense epochs do also contribute $ \mathcal{O}(p\sqrt{t}) $ work. Let us compare this with our solution.
 
 Noisy sparse epochs contribute $ \mathcal{O}(t) $ because these are phases with successful broadcasts. And there are clearly $ t $ tasks to perform, so at most $ t $ transmissions 
 will be necessary for this purpose.
 
 Silent sparse epochs, as well as dense epochs consist of mixed work: effective and failing. In our case, each attempt of transmitting is now simulated by
 $ \mathcal{O}(log^{2}(p)) $ rounds. That is why the amount of work is asymptotically multiplied by this factor. Hence we have work accrued during silent sparse and dense epochs contributing
 $ \mathcal{O}(p\sqrt{t}\log^{2}(p)) $.
 
 However according to Lemma \ref{lem42} with some small probability we could have mistakenly removed a group of stations from list \texttt{GROUPS} because \textsc{Mod-Confirm-Work} was silent. 
 Eventually the list of groups may be empty, and there are still some outstanding tasks.
 For such case we execute \textsc{Check-Outstanding}, where all the stations have the same outstanding tasks assigned, and do them for $ |\texttt{TASKS}| $ phases (which actually means until they are all done).
 It is clear that always at least one station remains operational and all the tasks will be performed. Work contributed in such case is at most $ \mathcal{O}(pt) $.
 
 Let us now estimate the expected work:
 
 $$ \left(1 - \frac{1}{p}\right)\mathcal{O}(t + p\sqrt{t}\log^{2}(p)) + \frac{1}{p}\mathcal{O}(pt) = \mathcal{O}(t + p\sqrt{t}\log^{2}(p)), $$
 what completes the proof.
\end{proof}

%% file: beeping.tex
\section{Transition to the beeping model}
\label{beeping}

To this point we considered a communication model based on a shared channel, with distinction that collision detection is not available. In this section we consider
the beeping model.

In the beeping model we distinguish two types of signals. One is silence, where no station transmits. The other is a beep, which, when heard, indicates that at least one station transmitted.
It differs from the channel with collision detection by providing slightly different feedback, but as we show it has the same complexity with respect to
reliable \DA.
More precisely, we show that the feedback provided by the beeping channel
allows to execute algorithm \textsc{Groups-Together} \jmii{(cf. Section \ref{grotog})} and that 
it is work optimal as well.

\subsection{Lower bound}

We state the lower bound for \DA\ in the beeping model in the following lemma. 

\begin{lemma}
A reliable algorithm, possibly randomized, with the beeping communication model performs work $ \Omega(t + p\sqrt{t}) $ in an execution in which no failures occur.
 \label{lem31}
\end{lemma}

\begin{proof}
The proof is an adaptation of the proof of Lemma 1 from \cite{CKL} to the beeping model.
Let $ \mathcal{A} $ be a reliable algorithm. The part $ \Omega(t) $ of the bound follows from the fact that every task has to be performed at least once in any execution of $ \mathcal{A} $.
 
 Task $ \alpha $ is \textit{confirmed} at round $ i $ of an execution of algorithm $ \mathcal{A} $, if either a station performs a beep successfully and it has performed $ \alpha $ by round $ i $, or
 at least two stations performed a beep simultaneously and all of them have performed task $ \alpha $ by round $ i $ of the execution.
 All of the stations broadcasting at round $ i $ and confirming $ \alpha $ have performed it by then, so at most $ i $ tasks can be confirmed at round $ i $. Let $ \mathcal{E}_{1} $
 be an execution of the algorithm when no failures occur. Let station $ v $ come to a halt at some round $ j $ in $ \mathcal{E}_{1} $.
 
\noindent
\textbf{Claim:} 
The tasks not confirmed by round~$ j $ were performed by $ v $ itself in $ \mathcal{E}_{1} $.

\begin{proof}
Suppose, to the contrary, that this is not the case, and let $ \beta $ be such a task.
Consider an execution, say $ \mathcal{E}_{2} $, obtained by running the algorithm and
crashing  any station that performed task $ \beta $ in $ \mathcal{E}_{1} $ just before it was to perform $ \beta $ in $ \mathcal{E}_{1} $, and all the remaining stations, except for $ v $,
crashed at step $ j $.
The broadcasts on the channel are the same during the first $ j $ rounds in $ \mathcal{E}_{1} $ and $ \mathcal{E}_{2} $.
Hence all the stations perform the same tasks in $ \mathcal{E}_{1} $ and $ \mathcal{E}_{2} $ till round $ j $.
The definition of $\mathcal{E}_{2} $ is consistent with the power of the \textit{Unbounded} adversary.
The algorithm is not reliable because task $ \beta $ is not performed in 
$ \mathcal{E}_{2} $ and station $ v $ is operational.
This justifies the claim.
\end{proof}

We estimate the contribution of the station $ v $ to work.
The total number of tasks confirmed in $ \mathcal{E}_{1} $ is at most 
\[
1+2+\ldots+j=\mathcal{O}(j^2)\ .
\]
Suppose some $ t' $ tasks have been confirmed by round $ j $.
The remaining $ t-t' $ tasks have been performed by $ v $.
The work of $ v $ is at least 
\[
\Omega(\sqrt{t'}+(t-t'))=\Omega(\sqrt{t})\ ,
\]
which completes the proof.
\end{proof}

\subsection{How algorithm \textsc{Groups-Together} works in the beeping model}

Collision detection was a significant part of algorithm \textsc{Groups-Together} as it provided the possibility of taking advantage of simultaneous transmissions. Because of maintaining common knowledge about
the tasks assigned to groups of stations we were not interested in the content of the transmission but the fact that at least one station from the group remained operational, what
guaranteed progress.

In the beeping model we cannot distinguish between \textit{Single} and \textit{Collision}, however in the sense of detecting progress the feedback is consistent.
It means that if a group $ g $ is scheduled to broadcast at some phase $ i $, then we have two possibilities. If \textit{Silence} was heard this means that all the stations
in group $ g $ were crashed, and their tasks remain outstanding. Otherwise if a beep is heard this means that at least one station in the group remained operational. As the transmission
was scheduled in phase $ i $ this means that certain $ i $ tasks were performed by group $ g $.

Lemma \ref{lem31} together with the work performance of \textsc{Groups-Together} allows us to conclude that the solution is also optimal in the beeping model.

\begin{corollary}
\jmii{\textsc{Groups-Together}} is work optimal in the beeping channel against the $ f $-Bounded adversary.
\end{corollary}

%% file: conclusions.tex
\section{Conclusions}
\label{conclusions}

\mj{
This paper addressed the challenge of performing work on a shared channel
with crash-prone stations against ordered \dk{and delayed} adversaries, introduced in this work.
The considered model is very basic, therefore our solutions could
be implemented and efficient in other related communication models
with contention and failures.

We found that some orders of crash events are more costly than the others for 
\dk{a given
algorithm and the whole problem,}
in particular, more shallow orders or even slight delays \jm{in the effects of adversary's decisions},
constraining the adversary, allow solutions to stay closer to the absolute lower bound
for this problem.

\remove{
In this paper our aim was to design reliable and most efficient algorithms
using randomization techniques against different adversarial scenarios.
}

All our algorithms work on a shared channel with acknowledgments only,
without collision detection,
what makes the setup challenging. While it was already shown that there
is not much we can do against a \safba\ adversary,
our goal was to 
investigate whether there are some other adversaries that 
an algorithm can play against
efficiently.

Taking a closer look at our algorithms, each of them works differently against
different adversaries.
\Robal\ 
does not simulate a collision detection
mechanism, opposed to the other two solutions, but tries to exploit good 
properties of an existing (but a priori unknown to the algorithm) linear order of crashes. On the other hand, its execution against a \walba\ adversary could be
inefficient --- the adversary could enforce a significant increase in the overall work performance \jm{by crashing a small number of stations multiple times. Then \textsc{Mix-And-Test} would be
executed many times with the same parameters, generating excessive work}.
\Grubtech, on the other hand,
cannot work efficiently against the  
1-RD 
adversary, as there is
a global leader chosen to coordinate the \textsc{Crash-Echo} procedure that simulates confirmations in a way similar to a collision detection mechanism (recall that
we do not assume collision detection given as channel feedback). 
Hence such an adversary could decide to always crash the leader, making the algorithm inefficient, as electing a leader is quite costly \jm{--- the leader is chosen in a number of trials, what generates excessive work}.
%
Yet, from a different angle, \Gilet\ 
confirms every piece of progress by electing a leader in a specific way, 
which is efficient against the $1$-RD adversary, but
executing it against the 
Weakly-Adaptive
adversary would result in an increase in the overall work complexity.

\input{timeenergy-short}

\paragraph{Open problems.}
Further study of distributed problems and systems 
against ordered adversaries seems to be a natural future direction.
Another interesting area is to study various extensions of the \DA\ problem in the shared-channel setting, such as considering a dynamic model, where additional tasks may appear while algorithm execution,
partially ordered sets of tasks, or tasks with different lengths and deadlines. 
In other words, to develop scheduling theory on a shared channel prone to failures.
In all the above mentioned directions, including the one considered in this work,
one of the most fundamental questions arises: Is there a universally efficient
solution against the whole range of adversarial scenarios?
\dk{Different natures of adversaries and properties of algorithms discussed above
suggest that it may be difficult to design such a universally efficient algorithm.}

\remove{
Nevertheless, it is still interesting to tackle logarithmic factors in the upper bounds shown here.
}

}

%% file: timeenergy-short.tex
\dk{
\paragraph{Remarks on time complexity.}
\label{timeenergy}

\remove{
In the presented problem and model, work complexity describes algorithms' performance quite precisely, however we would like to state \jm{general and less formal ideas} about time complexity and energy
consumption of our algorithms, \dk{the latter} understood as the total number of transmissions by stations. 
}

First of all, we emphasize that time complexity,
defined as the number of rounds until all non-crashed stations terminate, is not the best choice to describe how efficient the algorithms are, because this strongly depends on how the adversary interferes with the system.
In what follows we present some general bounds that might, however, overestimate the time complexity for a vast range of executions. 
\remove{
Clementi et al. \cite{CMS} proved that
%
any $ F $-reliable Do-All protocol, in the worst case, requires at least $ \Omega\left(\frac{t}{p-F} + \min\left\{\frac{tF}{p}, F + \sqrt{t}\right\}\right)$ completion time.
  It holds even when the faults only happen at the very beginning of the protocol execution.
The also showed an $ F $-reliable Do-All protocol 
having optimal completion time.
%
}
\jm{In all our considerations, at some point of an execution (even at the very beginning) it may happen that only the non-faulty stations remain operational, 
because the adversary will realize all of its possible crashes. Then at most $ t $ tasks must be performed by the remaining $ p - f $ stations. Hence, even if the tasks are equally
distributed among the non-faulty stations, doing them all lasts at least $ t/(p-f) $ rounds.
On the other hand, initially $ t $ tasks are distributed among $ p $ stations. Thus, on average, a station will be working on $ t/p $ tasks. If now the adversary decides to crash a station
just before it was to confirm its tasks, then this prolongs the overall execution by $ t/p $ rounds. Because there are $ f $ crashes, then at most $ tf/p $ rounds are additionally needed to finish.
However, it is also true that stations are capable of performing $ t $ tasks in $ \sqrt{t} $ rounds. This corresponds to the triangular \textsc{Two-Lists}-fashion of assigning tasks to stations.
In this view each crash enforces an additional step of the execution, what gives us the upper bound of around $ f + \sqrt{t} $ rounds.

All our algorithms undergo the same time bounds for actually performing tasks or suffering crashes as mentioned above. Additionally, $ \mathcal{O}(\sqrt{t}\log(p)) $ rounds are needed for \textsc{ROBAL} to select sets of leaders
throughout all the executions of the \textsc{Mix-And-Test} procedure. Consequently,  the expected running time of \textsc{ROBAL} is 
$\mathcal{O}\left(\frac{t}{p-f} + \min\left\{\frac{tf}{p}, f + \sqrt{t}\right\} + \sqrt{t}\log(p) \right)$.
Following the same reasoning, \textsc{GrubTEch} algorithm, apart from doing productive work in the presence of the adversary, will require additional time for the leader election mode, which is 
$\mathcal{O}\left(\frac{p}{p-f}\log(p)\right)$ in expectation.
The total expected running time of \textsc{GrubTEch} is therefore $\mathcal{O}\left(\frac{t}{p-f} + \min\left\{\frac{tf}{p}, f + \sqrt{t}\right\} + \frac{p}{p-f}\log(p) \right)$.
In \textsc{GILET} each transmission is confirmed by electing a leader, hence its expected running time is $\mathcal{O}\left(\left(\frac{t}{p-f} + \min\left\{\frac{tf}{p}, f + \sqrt{t}\right\}\right)\log^{2}(p)\right)$.

%
%

\paragraph{Remarks on energy complexity.}

Since our algorithms are randomized, it is also quite difficult to state tight bounds for the transmission energy used in executions.
Here by transmission energy we understand the total number of transmissions undertaken by stations during the execution.
Nevertheless, assuming that $ n $ denotes the number of operational stations
and there is a certain amount of work $ S $ accrued by some time of an execution of any of our algorithms, then $ S/\sqrt{n} $ is roughly (the upper bound on) the number of transmissions 
done by that time. 
This is because substantial parts of our algorithms are based on 
procedure \textsc{Groups-Together}, in which
roughly $\sqrt{n'}$ stations in a group transmit
in a round, out of at least $n'\ge n$ operational ones that contribute
to the total work $S$. 
However, our algorithms also strongly rely on different leader election type of procedures, therefore the total transmission energy cost in an execution 
may vary significantly.


}
}